\documentclass[a4paper,UKenglish,cleveref, autoref, thm-restate]{lipics-v2021}

\usepackage{xspace}
\title{Efficient Algorithms for \ESMT on Near-Convex Terminal Sets} 

\titlerunning{ESMT on Near-Convex Terminal Sets} 

\author{Anubhav Dhar}{Indian Institute of Technology Kharagpur, India}{anubhavdhar@kgpian.iitkgp.ac.in}{}{}

\author{Soumita Hait}{Indian Institute of Technology Kharagpur, India}{soumitahait7321@gmail.com}{}{}

\author{Sudeshna Kolay}{Indian Institute of Technology Kharagpur, India}{skolay@cse.iitkgp.ac.in}{}{}

\authorrunning{A.~Dhar, S.~Hait, S.~Kolay} 

\Copyright{} 



\category{} 

\relatedversion{} 

\nolinenumbers 

\hideLIPIcs  

\EventEditors{Petr A. Golovach and Meirav Zehavi}
\EventNoEds{2}
\EventLongTitle{16th International Symposium on Parameterized and Exact Computation (IPEC 2021)}
\EventShortTitle{IPEC 2021}
\EventAcronym{IPEC}
\EventYear{2021}
\EventDate{September 8--10, 2021}
\EventLocation{Lisbon, Portugal}
\EventLogo{}
\SeriesVolume{214}
\ArticleNo{2}

\newcommand{\defproblem}[3]{
  \vspace{1mm}
\noindent\fbox{
  \begin{minipage}{0.96\textwidth}
  \begin{tabular*}{\textwidth}{@{\extracolsep{\fill}}lr} #1 \\ \end{tabular*}
  {\bf{Input:}} #2  \\
  {\bf{Question:}} #3
  \end{minipage}
  }
  \vspace{1mm}
}

\newcommand{\C}[1]{\mathcal{#1}}

\usepackage{todonotes}
\usepackage{algorithm}
\usepackage{algpseudocode}

\newcommand{\ESMT}{\textsc{Euclidean Steiner Minimal Tree}\xspace}
\newcommand{\pnp}{$\mbox{P} = \mbox{NP}$\xspace}

\newcommand{\OO}{\mathcal{O}}

\newcommand{\smtpoly}{SMT for $\{A_i\} \cup \{B_i\}$ \xspace}

 \definecolor{babyblue}{rgb}{0.54, 0.81, 0.94}
 \definecolor{b1}{rgb}{0.63, 0.79, 0.95}
 \definecolor{b2}{rgb}{0.74, 0.83, 0.9}
 \definecolor{b3}{rgb}{0.67, 0.9, 0.93}
 \definecolor{gentlegreen}{rgb}{0.00, 0.51, 0.00}

\begin{document}

\maketitle

\begin{abstract}
    The \ESMT problem takes as input a set $\mathcal P$ of points in the Euclidean plane and finds the minimum length network interconnecting all the points of $\mathcal P$. In this paper, in continuation to the works of~\cite{du1987steiner} and ~\cite{weng1995steiner}, we study \ESMT when $\mathcal P$ is formed by the vertices of a pair of regular, concentric and parallel $n$-gons. 
    We restrict our attention to the cases where the two polygons are not very close to each other. In such cases, we show that \ESMT is polynomial-time solvable, and we describe an explicit structure of a Euclidean Steiner minimal tree for $\mathcal P$.
    
    We also consider point sets $\mathcal P$ of size $n$ where the number of input points not on the convex hull of $\mathcal P$ is $f(n) \leq n$. We give an exact algorithm with running time $2^{\OO(f(n)\log n)}$ for such input point sets $\mathcal P$. Note that when $f(n) = \OO(\frac{n}{\log n})$, our algorithm runs in single-exponential time, and when $f(n) = o(n)$ the running time is $2^{o(n\log n)}$ which is better than the known algorithm in~\cite{hwang1992steiner}. 
    
    We know that no FPTAS exists for \ESMT unless \pnp~\cite{garey1977complexity}. On the other hand FPTASes exist for \ESMT on convex point sets~\cite{scott1988convexity}. In this paper, we show that if the number of input points in $\mathcal P$ not belonging to the convex hull of $\mathcal P$ is $\OO(\log n)$, then an FPTAS exists for \ESMT. In contrast, we show that for any $\epsilon \in (0,1]$, when there are $\Omega(n^{\epsilon})$ points not belonging to the convex hull of the input set, then no FPTAS can exist for \ESMT unless \pnp. 
\end{abstract}

\section{Introduction}\label{sec:intro}

The \ESMT problem asks for a network of minimum total length interconnecting a given finite set $\mathcal P$ of $n$ points in the Euclidean plane. Formally, we define the problem as follows, taken from~\cite{brazil2014history}:

\defproblem{\ESMT}{A set $\mathcal P$ of $n$ points in the Euclidean plane}{Find a connected plane graph $\mathcal T$ such that $\mathcal P$ is a subset of the vertex set $V(\mathcal T)$, and for the edge set $E(\mathcal T)$, $\Sigma_{e \in E(\mathcal T)} \overline e$ is minimized over all connected plane graphs with $\mathcal P$ as a vertex subset.}

Note that the metric being considered is the Euclidean metric, and for any edge $e \in E(\mathcal T)$, $\overline e$ denotes the Euclidean length of the edge. Here, the input set $\C{P}$ of points is often called a set of {\it terminals}, the points in $\C{S} = V(\mathcal T) \setminus \C{P}$ are called {\it Steiner points}. A solution graph $\mathcal T$ is referred to as a {\it Euclidean Steiner minimal tree}, or simply an SMT.

 The \ESMT problem is a classic problem in the field of Computational Geometry. The origin of the problem dates back to Fermat (1601-1665) who proposed the problem of finding a point in the plane such that the sum of its distance from three given points is minimized. This is equivalent to finding the location of the Steiner point when given three terminals as input. Torricelli proposed a geometric solution to this special case of 3 terminal points. The idea was to construct equilateral triangles outside on all three sides of the triangle formed by the terminals, and draw their circumcircles. The three circles meet at a single point, which is our required Steiner point. This point came to be known as the {\it Torricelli point}. When one of the angles in the triangle is at least $120^\circ$, the minimizing point coincides with the obtuse angle vertex of the triangle. In this case, the Torricelli point lies outside the triangle and no longer minimizes the sum of distances from the vertices. However, when vertices of polygons with more than 3 sides are considered as a set of terminals, a solution to the Fermat problem does not in general lead to a solution to the \ESMT problem. For a more detailed survey on the history of the problem, please refer to~\cite{brazil2014history,hwang1992steiner}. For convenience, we refer to the \ESMT problem as ESMT.
 

 ESMT is NP-hard. In~\cite{garey1977complexity}, Garey et al.~prove a discrete version of the problem (Discrete ESMT) to be strongly NP-complete via a reduction from the \textsc{Exact Cover by 3-Sets} (X3C) problem. Although it is not known if the ESMT problem is in NP, it is at least as hard as any NP-complete problem. So, we do not expect a polynomial time algorithm for it. A recursive method using only Euclidean constructions was given by Melzak in~\cite{melzak1961problem} for constructing all the Steiner minimal trees for any set of $n$ points in the plane by constructing full Steiner trees of subsets of the points. Full Steiner trees are interconnecting trees having the maximum number of newly introduced points (Steiner points) where all internal junctions are of degree $3$. Hwang improved the running time of Melzak’s original exponential algorithm for full Steiner tree construction to linear time in~\cite{hwang1986linear}. Using this, we can construct an Euclidean Steiner minimal tree in $2^{\OO(n\log n)}$ time for any set of $n$ points. This was the first algorithm for \ESMT. The problem is known to be NP-hard even if all the terminals lie on two parallel straight lines, or on a bent line segment where the bend has an angle of less than $120^\circ$~\cite{rubinstein1997steiner}. Since the above sets of terminals all lie on the boundary of a convex polygon (or, are in convex position), this shows that ESMT is NP-hard when restricted to a set of points that are in weakly convex position.

 Although the ESMT problem is NP-hard, there are certain arrangements of points in the plane for which the Euclidean Steiner minimal tree can be computed efficiently, say in polynomial time. One such arrangement is placing the points on the vertices of a regular polygon. This case was solved by Du et al.~\cite{du1987steiner}. Their work gives exact topologies of the Euclidean Steiner minimal trees. Weng et al.~\cite{weng1995steiner} generalized the problem by incorporating the centre point of the regular polygon as part of the terminal set, along with the vertices. This case was also found to be polynomial time solvable.

Tractability in the form of approximation algorithms for ESMT has been extensively studied. It was proved in~\cite{garey1977complexity} that a fully polynomial time approximation scheme (FPTAS) cannot exist for this problem unless \pnp. However, we do have an FPTAS when the terminals are in convex position~\cite{scott1988convexity}. Arora's celebrated polynomial time approximation scheme (PTAS) for the ESMT and other related problems is described in~\cite{arora1998polynomial}. Around the same time, Rao and Smith gave an efficient polynomial time approximation scheme (EPTAS) in~\cite{rao1998approximating}. In recent years, an EPTAS with an improved running time was designed by Kisfaludi-Bak et al.~\cite{kisfaludi2022gap}. 
\paragraph*{Our Results.}

 In this paper, we first extend the work of~\cite{du1987steiner}  and~\cite{weng1995steiner}.
We state this problem as ESMT on $k$-Concentric Parallel Regular $n$-gons.

 \begin{definition}[$k$-Concentric Parallel Regular $n$-gons] \label{def:regular conc parallel polygons toh}
     $k$-Concentric Parallel Regular $n$-gons are $k$ regular $n$-gons that are concentric and where the corresponding sides of polygons are parallel to each other. 
 \end{definition}

 Please refer to Figure~\ref{fig:examples}(a) for an example of a $2$-Concentric Parallel Regular $12$-gon. We call $k$-Concentric Parallel Regular $n$-gons as $k$-CPR $n$-gons for short.

 We consider terminal sets where the terminals are placed on the vertices of $2$-CPR $n$-gons. In the case of $k=2$, the $n$-gon with the smaller side length will be called the inner $n$-gon and the other $n$-gon will be called the outer $n$-gon. Also, let $a$ be the side length of the inner $n$-gon, and $b$ be the side length of the outer $n$-gon. We define $\lambda = \frac{b}{a}$ and refer to it as the \emph{aspect ratio} of the two regular polygons. In~\Cref{polycase}, we derive the exact structures of the SMTs for $2$-CPR $n$-gons when the aspect ratio $\lambda$ of the two polygons is greater than $\frac{1}{1-4\sin{(\pi/n)}}$ and $n \geq 13$.

 Next, we consider ESMT on an $f(n)$-Almost Convex Point Set.
 \begin{definition}[$f(n)$-Almost Convex Point Set] \label{def:almost_conv}
     An $f(n)$-Almost Convex Point Set $\mathcal P$ is a set of $n$ points in the Euclidean plane such that there is a partition $\mathcal P = \mathcal P_1 \uplus \mathcal P_2$ where $\mathcal P_1$ forms the convex hull of $\mathcal P$ and $|\mathcal P_2| = f(n)$.
 \end{definition}
Please refer to Figure~\ref{fig:examples}(b) for an example of a $5$-Almost Convex Set of $13$ points. In~\Cref{sec:exact_algo}, we give an exact algorithm for ESMT on $f(n)$-Almost Convex Sets of $n$ terminals. The running time of this algorithm is $2^{\OO(f(n) \log n)}$. Thus, when $f(n) = \OO(\frac{n}{\log n})$, then our algorithm runs in $2^{\OO(n)}$ time, and when $f(n) = o(n)$ then the running time is $2^{o(n\log n)}$. This is an  improvement on the best known algorithm for the general case~\cite{hwang1992steiner}.

Next, for $f(n) = \OO(\log n)$, we give an FPTAS in~\Cref{subsec:fptas}. On the other hand in~\Cref{subsec:apx_hardness} we show that, for all $\epsilon \in (0,1]$, when $f(n) \in \Omega(n^\epsilon)$, there cannot exist any FPTAS unless \pnp.

\begin{figure}[h]
\centering
\subfloat[\centering $2$-CPR $12$-gons] {\includegraphics[width=6cm]{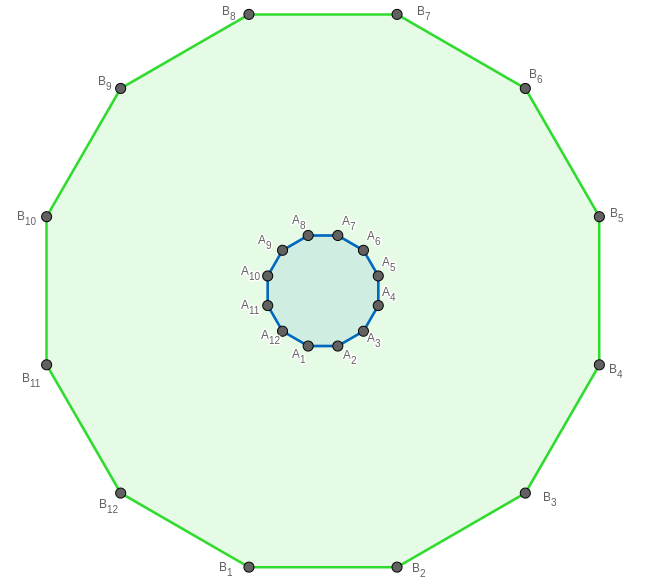}}
    \qquad
\subfloat[\centering \centering $f(n)$-Almost Convex Point Set for $n = 13$, $f(n) = 5$] {\includegraphics[width=6cm]{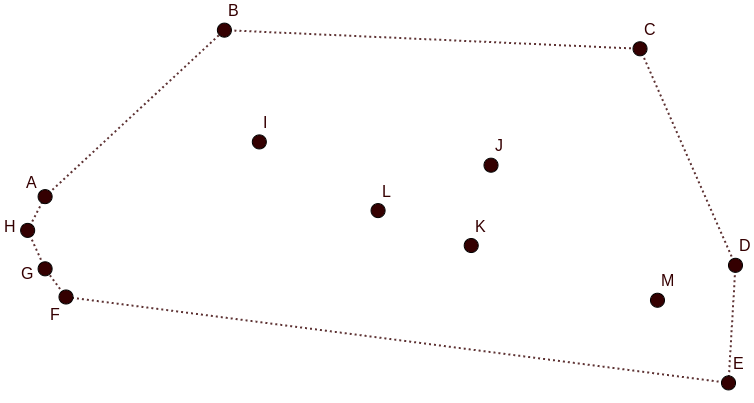}}
\caption{Examples for~\Cref{def:regular conc parallel polygons toh} and~\Cref{def:almost_conv}}
\end{figure}\label{fig:examples}

\section{Preliminaries}\label{sec:prelims}

\subparagraph{Notations.}
For a given positive integer $k \in \mathbb{N}$, the set of integers $\{1,2,\ldots,k\}$ is denoted for short as $[k]$. Given a graph $G$, the vertex set is denoted as $V(G)$ and the edge set as $E(G)$. Given two graphs $G_1$ and $G_2$, $G_1 \cup G_2$ denotes the graph $G$ where $V(G) = V(G_1) \cup V(G_2)$ and $E(G) = E(G_1) \cup E(G_2)$. 

In this paper, a regular $n$-gon is denoted by $A_1A_2A_3...A_n$ or $B_1B_2B_3...B_n$. For convenience, we define $A_{n + 1} := A_1$, $B_{n + 1} := B_1$, $A_0 := A_n$ and $B_0 := B_n$. We use the notation $\{A_i\}$ to denote the polygon $A_1A_2A_3 \ldots A_n$ and $\{B_i\}$ to denote the polygon $B_1B_2B_3 \ldots B_n$. For any regular polygon $A_1A_2A_3...A_n$, the circumcircle of the polygon is denoted as $(A_1A_2A_3...A_n)$. Given any $n$-vertex polygon in the Euclidean plane with vertices $\mathcal P = P_1P_2P_3\ldots P_n$, and interval in $\mathcal{K}$ is a subset of consecutive vertices $P_iP_{i+1\ldots P_j}$, $i,j\in [n]$, also denoted as $[P_i,P_j]$. Here $P_i$ is considered the starting vertex of the interval and $P_j$ the ending vertex. For any $P_k$, $i \leq k\leq j$ in the interval we will also use the notation $P_i \leq P_k \leq P_j$.

Given two points $P$, $Q$ in the Euclidean plane, we denote by ${\sf dist}(P,Q)$ the Euclidean distance between $P$ and $Q$. Given a line segment $AB$ in the Euclidean plane, $\overline{AB} = {\sf dist}(A,B)$. For two distinct points $A$ and $B$, $L_{AB}$ denotes the line containing $A$ and $B$; and $\overrightarrow{AB}$ denotes the ray originating from $A$ and containing $B$. 

When we refer to a graph $\mathcal{G}$ in the Euclidean plane then $V(\mathcal{G})$ is a set of points in the Euclidean plane, and $E(\mathcal{G})$ is a subset of the family of line segments $\{P_1P_2 | P_1,P_2 \in V(\mathcal{G})\}$. For any tree $\mathcal T$ in the Euclidean plane, we denote by the notation $|\mathcal T|$ the value of $\Sigma_{e \in E(\mathcal T)} \overline{e}$. A path in a tree $\mathcal T$ is uniquely specified by the sequence of vertices on the path; therefore, $P_1$, $P_2$, $P_3$, \ldots, $P_k$ (where $P_i \in V(\mathcal T), \forall i \in [k]$ and $P_iP_{i+1} \in E(\mathcal T), \forall i \in [k-1]$) denotes the path starting from the vertex $P_1$, going through the vertices $P_2$, $P_3$, \ldots, $P_{k-1}$ and finally ending at $P_k$. Equivalently, we can specify the same path as \emph{the path from $P_1$ to $P_k$}, since $\mathcal T$ is a tree. Consider the graph $T$ such that $V(T) = \{v_P| P \in V(\mathcal{T})\}$, $E(T) = \{v_{P_1}v_{P_2}| P_1P_2 \mbox{ is a line segment in } E(\mathcal{T})\}$. Then $T$ is said to be the topology of $\mathcal{T}$ while $\mathcal{T}$ is said to realize the topology $T$. Given two trees $\mathcal{T}_1$, $\mathcal{T}_2$ in the Euclidean plane, $\mathcal{T}' = \mathcal{T}_1\cup \mathcal{T}_2$ is the graph where $V(\mathcal{T}')= V(\mathcal{T}_1) \cup V(\mathcal{T}_2)$ and $E(\mathcal{T}')= E(\mathcal{T}_1) \cup E(\mathcal{T}_2)$. 

Given any graph $G$, a Steiner minimal tree or SMT for a terminal set $\mathcal{P} \subseteq V(G)$ is the minimum length connected subgraph $G'$ of $G$ such that $\mathcal{P} \subseteq V(G')$. The {\sc Steiner Minimal Tree} problem on graphs takes as input a set $\mathcal{P}$ of terminals and aims to find a minimum length SMT for $\mathcal{P}$. For the rest of the paper, we also refer to a Euclidean Steiner minimal tree as an SMT. Given a set of points $\mathcal{P}$ in the Euclidean plane, the convex hull of $\mathcal{P}$ is denoted as $\mathrm{CH(\mathcal P)}$.

\subparagraph{Euclidean Minimum Spanning Tree (MST).}
Given a set $\mathcal P$ of $n$ points in the Euclidean plane, let $G$ be a graph where $V(G) = \{v_P| P \in \mathcal P\}$ and $E(G) = \{v_{P_i}v_{P_j} | P_i,P_j \in \mathcal P\}$. Also, a weight function $w_{G}: E(T) \rightarrow \mathbb{R}$ is defined such that for each edge $v_{P_1}v_{P_2} \in E(T)$, $w_{G}(v_{P_1}v_{P_2}) = \overline{P_1P_2}$. The Euclidean minimum spanning tree of a set $\mathcal P$ is the minimum spanning tree of the graph $G$ with edge weights $w_G$. Note that a Steiner tree may have shorter length than a minimum spanning tree of the point set $\mathcal P$. 

In the plane, the Euclidean minimum spanning tree is a subgraph of the Delaunay triangulation. Using this fact, the Euclidean minimum spanning tree for a given set of points in the Euclidean plane can be found in $\OO(n\log n)$ time as discussed in \cite{Shamos1975ClosestpointP}. 

\subparagraph{Properties of a Euclidean Steiner minimal tree.}
A Euclidean Steiner minimal tree (SMT) has certain structural properties as given in~\cite{cockayne1967steiner}. We state them in the following Proposition.

\begin{proposition}\label{smt-prop}
Consider an SMT on $n$ terminals.
 \begin{enumerate}
   \item No two edges of the SMT intersect with each other.
 
   \item Each Steiner point has degree exactly $3$ and the incident edges meet at $120^\circ$ angles. The terminals have degree at most $3$ and the incident edges form angles that are at least $120^\circ$.
  
   \item The number of Steiner points is at most $n-2$, where $n$ is the number of terminals.

\end{enumerate}
\end{proposition}

 A full Steiner tree (FST) is a Steiner tree (need not be minimal, but may include Steiner points) having exactly $n-2$ Steiner points, where $n$ is the number of terminals. In an FST, all terminals are leaves and Steiner points are interior nodes. When the length of an FST is minimized, it is called a minimum FST.

All SMTs can be decomposed into FST components such that, in each component a terminal is always a  leaf. This decomposition is unique for a given SMT~\cite{hwang1992steiner}. A topology for an FST is called a full Steiner topology and that of a Steiner tree is called a Steiner topology.


\subparagraph{Steiner Hulls.}
A Steiner hull for a given set of points is defined to be a region which is known to contain an SMT. We get the following propositions from~\cite{hwang1992steiner}.

\begin{proposition}\label{convex-steiner}
    For a given set of terminals, every SMT is always contained inside the convex hull of those points. Thus, the convex hull is also a Steiner hull.
\end{proposition}

The next two propositions are useful in restricting the structure of SMTs and the location of Steiner points.

\begin{proposition} [The Lune property]\label{lune}
    Let $\rm UV$ be any edge of an SMT. Let $L(\rm{U},\rm{V})$ be the lune-shaped intersection of circles of radius $|\rm UV|$ centered on $\rm U$ and $\rm V$. No other vertex of the SMT can lie in $L(\rm{U},\rm{V})$, except $U$ and $V$ themselves.
\end{proposition}

\begin{proposition} [The Wedge property]\label{wedge}
    Let $W$ be any open wedge-shaped region having angle $120^\circ$ or more and containing none of the points from the input terminal set $\mathcal P$. Then $W$ contains no Steiner points from an SMT of $\mathcal P$.
\end{proposition}

\subparagraph{Approximation Algorithms.}
We define all the necessary terminology required in terms of a minimization problem, as ESMT is a minimization problem.


\begin{definition} [Efficient Polynomial Time Approximation Scheme (EPTAS)]
    An algorithm is called an efficient polynomial time approximation scheme (EPTAS) for a problem if it takes an input instance and a parameter $\epsilon > 0$, and outputs a solution with approximation factor $(1+\epsilon)$ for a minimization problem in time $f(1/\epsilon)n^{\OO(1)}$ where $n$ is the input size and $f(1/\epsilon)$ is any computable function.
\end{definition}

\begin{definition} [Fully Polynomial Time Approximation Scheme (FPTAS)]
    An algorithm is called a fully polynomial time approximation scheme (FPTAS) for a problem if it takes an input instance and a parameter $\epsilon > 0$, and outputs a solution with approximation factor $(1+\epsilon)$ for a minimization problem in time $(1/\epsilon)^{\OO(1)}n^{\OO(1)}$ where $n$ is the input size.
\end{definition}


\section{Polynomial cases for \ESMT}\label{polycase}
In this section, we consider the \ESMT problem for $2$-CPR $n$-gons. Throughout the section, we denote the inner $n$-gon as $\{A_i\}$ and the outer $n$-gon as $\{B_i\}$. First, we consider the configuration of an Euclidean Steiner minimal tree in a subsection of the annular area between $\{A_i\}$ and $\{B_i\}$, which will form an isosceles trapezoid. Next, we consider the simple but illustrative case of $n = 3$. Finally we prove our result for all $n$.

\subsection{Isosceles Trapezoids and Vertical Forks} \label{trapezoids}

In this section, we discuss one particular Steiner topology when the terminal set is  formed by the four corners of a given isosceles trapezoid. However, we will limit the discussion to only the isosceles trapezoids such that the angle between the non-parallel sides is of the form $\frac{2 \pi}n$ where $n \in \mathbb{N}$, $n \ge 4$. The reason is that given $2$-CPR $n$-gons $\{A_i\}$, $\{B_i\}$, for $n \geq 4$ and for any $j \in \{1,\ldots, n-1\}$, the region $A_jA_{j + 1}B_jB_{j + 1}$ is an isosceles trapezoid such that the angle between the non-parallel sides is of the form $\frac{2 \pi}n$.

\begin{figure}[h]
\centering
\includegraphics[width=5.5cm]{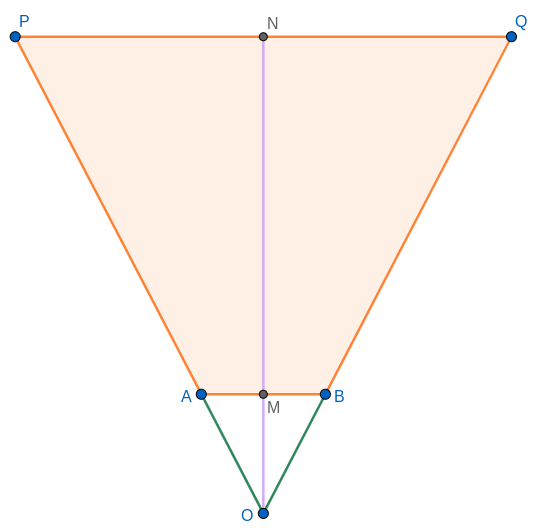}
\caption{Isosceles trapezoid with $\angle AOB = \frac {2 \pi} 8$}
\label{trap_def}
\end{figure}

Let $ABQP$ be an isosceles trapezoid with $AB$, $PQ$ as the parallel sides, and $AP$, $BQ$ as the non-parallel sides. Assume without loss of generality that $AB$ is shorter in length than $PQ$. Let $\overline{AB} = 1$ and $\overline{PQ} = \lambda$, where $\lambda \geq \frac {\sqrt 3 + \tan{ \frac \pi n }} {\sqrt 3 - \tan{ \frac \pi n}} $. For brevity, we say $\lambda_v = \frac {\sqrt 3 + \tan{ \frac \pi n }} {\sqrt 3 - \tan{ \frac \pi n}}$. Let $L_{PA}$ and $L_{QB}$ be the lines containing the line segments $PA$ and $QB$ respectively. Also let $O$ be the point of intersection of $L_{PA}$ and $L_{QB}$. Further, let $M$ and $N$ be the midpoints of $AB$ and $PQ$ respectively (as in Figure \ref{trap_def}). As mentioned earlier, $\angle AOB = \frac{2 \pi}n$ where $n \in \mathbb{N}$, $n \ge 4$.

Now, we explore the following Steiner topology of the terminal set $\{A, B, P, Q\}$:
\begin{enumerate}
\item $A$ and $B$ are connected to a Steiner point $S_1$.
\item $P$ and $Q$ are connected to another Steiner point $S_2$. 
\item $S_1$ and $S_2$ are directly connected (Please see Figure \ref{T_vf}). 
\end{enumerate}
We call such a topology a \emph{vertical fork topology} and the Steiner tree realising such a topology, the \emph{vertical fork}. Note that in a vertical fork topology the only unknowns are the locations of the two Steiner points $S_1,S_2$. Therefore, we have the vertical fork topology as $T_{vf}$, with $E(T_{vf}) = \{AS_1, BS_1, S_1S_2, S_2P, S_2Q\}$.

\begin{figure}[h]
\centering
\includegraphics[width=5.5cm]{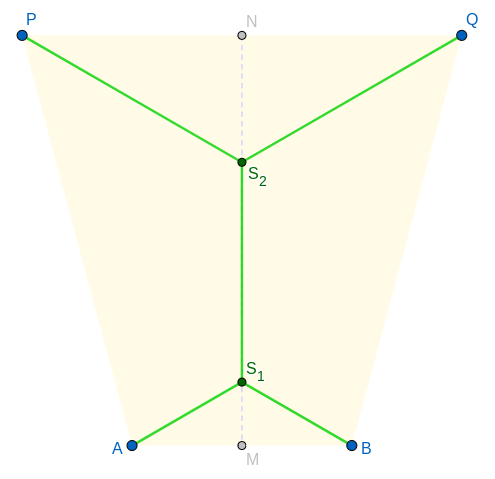}
\caption{The Vertical Fork, $\mathcal T_{vf}$}
\label{T_vf}
\end{figure}



We show the existence of a vertical fork and calculate its total length in the following lemma.

\begin{lemma} \label{lambda_v}
A vertical fork $\mathcal T_{vf}$ can be constructed for any $n \ge 4$ and for any $\lambda \ge \lambda_v$, where

$$\lambda_v = \frac {\sqrt 3 + \tan{ \frac \pi n }} {\sqrt 3 - \tan{ \frac \pi n}}$$    
such that the length of the vertical fork   
$$|\mathcal T_{vf}| = \dfrac{(\lambda - 1)}{2 \tan \frac \pi n} + \dfrac {\sqrt 3 (\lambda + 1)} {2}$$

\end{lemma}

\begin{proof}
First, we construct the Steiner points $S_1$, $S_2$ and then prove that the construction works.

In the following construction, we describe how to find the locations of $S_1$ and $S_2$ for  the vertical fork:
\begin{itemize}
    \item We construct equilateral triangles $ABE$ and $PQF$ where both points $E$ and $F$ lie outside the trapezoid $ABQP$.
    \item We construct the circumcircles $(ABE)$ and $(PQF)$ of $ABE$ and $PQF$, respectively. 
    \item Recall that $L_{MN}$ is the line segment containing $M$ and $N$. Define $S_1$ to be the point of intersection of $L_{MN}$ and the circle $(ABE)$ distinct from $E$; similarly, $S_2$ is the point of intersection of the $L_{MN}$ and $(PQF)$ distinct from $F$. Therefore, by construction, $S_2$, $M$ must lie on the same side of $N$ on $L_{MN}$, and $S_1$, $N$ must lie on the same side of $M$ on $L_{MN}$. Further, $\angle AS_1B = \angle PS_2Q = \frac{2 \pi}{3}$ by construction.
\end{itemize}

We now show that the points $S_1$ and $S_2$ indeed lie inside the line segment $MN$ and the points appear in the order: $M$, $S_1$, $S_2$, $N$. We prove the following claim to serve this purpose.

\begin{claim} \label{S1_S2_in_MN}
$\overline{S_1M} + \overline{S_2N} \le \overline{MN}$

\end{claim}
\begin{proof}
We have $\overline{MN} = \overline {ON} - \overline{OM} = \dfrac{\lambda}{2 \tan \frac{\pi}{n}} - \dfrac{1}{2 \tan \frac{\pi}{n}} = \dfrac{(\lambda - 1)}{2 \tan \frac \pi n}$\\
Again $\overline{S_1M} = \dfrac{1}{2\sqrt{3}}$ and $\overline{S_2N} = \dfrac{\lambda}{2\sqrt{3}}$.\\
Therefore we have 
\begin{align*}
   & \overline{MN} - (\overline{S_1M} + \overline{S_2M}) \\
 = & \dfrac{(\lambda - 1)}{2 \tan \frac \pi n} - \dfrac 1 {\sqrt{3}} - \dfrac \lambda {2\sqrt 3}\\
 = & \dfrac{(\lambda - 1)}{2 \tan \frac \pi n} - \dfrac {(\lambda + 1)} {2\sqrt{3}}\\
 = & \frac{\lambda(\sqrt{3} - \tan \frac{\pi}{n}) - (\sqrt{3} + \tan \frac{\pi}{n})}{2 \sqrt{3} \tan \frac \pi n}\\
 = & \frac{\sqrt{3} - \tan \frac{\pi}{n}}{2 \sqrt{3} \tan \frac \pi n} \cdot (\lambda - \lambda_v)
\end{align*}

Therefore $\lambda \ge \lambda_v$ implies $\overline{MN} \ge (\overline{S_1M} + \overline{S_2M})$. This proves~\Cref{S1_S2_in_MN}.
\end{proof}

From~\Cref{S1_S2_in_MN}, we get $\overline{S_1M}, \overline{S_2N} \le \overline{MN}$. As $S_2$, $M$ lie on the same side of $N$ on $L_{MN}$ and $S_1$, $N$ lie on the same side of $M$ on $L_{MN}$, this implies that $S_1$, $S_2$ lie on the line segment $MN$. Further, $\overline{S_1M} + \overline{S_2N} \le \overline{MN}$ implies $\overline{S_1M} \le \overline {S_2M} \le \overline{NM}$, which in turn implies that the points appear in the order: $M$, $S_1$, $S_2$. $N$.\\

Now, we calculate the total length of the vertical fork, $|\mathcal T_vf|$:
\begin{align*}
  & |\mathcal T_{vf}|\\
= & \overline{AS_1} + \overline{BS_1} + \overline{S_1S_2} + \overline{PS_2} + \overline{QS_2}\\
= & 2 \overline{PS_2} + 2 \overline{AS_1} + \overline{S_1S_2}\\
= & \dfrac{2 \lambda} {\sqrt 3} + \dfrac 2 {\sqrt 3} + \bigg( \dfrac{(\lambda - 1)}{2 \tan \frac \pi n} - \dfrac {(\lambda + 1)} {2\sqrt{3}} \bigg)\\
= & \dfrac{(\lambda - 1)}{2 \tan \frac \pi n} + \dfrac {\sqrt 3 (\lambda + 1)} {2}
\end{align*}













This completes the proof of~\Cref{lambda_v}.
\end{proof}

\subsection{\ESMT for $2$-Concentric Parallel Regular $3$-gons} \label{triangles}

Note that a regular $3$-gon is an equilateral triangle and therefore, for the rest of this section we call a regular $3$-gon as an equilateral triangle. We describe a minimal solution for \ESMT for $2$-CPR equilateral triangles.

\begin{lemma}\label{3-gon}
Consider two concentric and parallel equilateral triangles $A_1A_2A_3$ and $B_1B_2B_3$, where $B_1B_2B_3$ has side length $\lambda > 1$ and $A_1A_2A_3$ has side length $1$. An SMT of $\{A_1, A_2, A_3, B_1, B_2, B_3\}$ is an SMT of $\{B_1,B_2,B_3\}$, and has length $\sqrt 3 \cdot \lambda$. 
\end{lemma}

\begin{proof}
It is to be noted that the centre $O$ of both $A_1A_2A_3$ and $B_1B_2B_3$ is also the Torricelli point of both $A_1A_2A_3$ and $B_1B_2B_3$~\cite{du1987steiner}. On taking $O$ as the only Steiner point, the SMT for $\{B_1,B_2,B_3\}$ is $ \mathcal T_3 = \{OB_1, OB_2, OB_3\}$~\cite{du1987steiner}. However, the edges of $\mathcal T_3$ already pass through $A_1$, $A_2$ and $A_3$. Therefore, $ \mathcal T_3$ with $E(\mathcal T_3) = \{OA_1, OA_2, OA_3, A_1B_1, A_2B_2, A_3B_3\}$ is also the SMT for $\{A_1, A_2, A_3, B_1, B_2, B_3\}$ as shown in Figure \ref{conc_eq_tri}.

 From the definition of $\mathcal T_3$, we have the length of the \smtpoly, for $n = 3$ as 
$$ | \mathcal T_3 | = \sqrt 3 \cdot \lambda $$ 
\end{proof}

\begin{figure}[h]
\centering
\includegraphics[width=4cm]{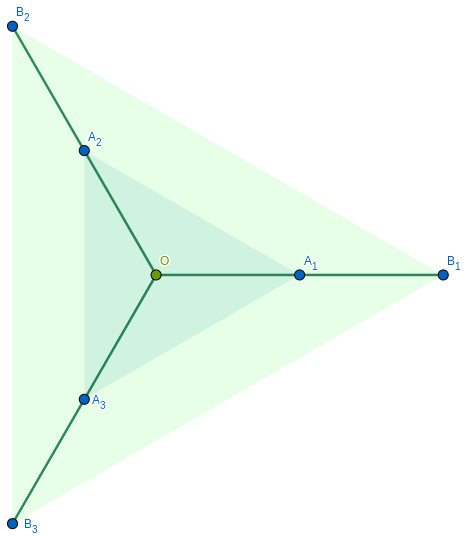}
\caption{SMT for $2$-CPR $3$-gons }
\label{conc_eq_tri}
\end{figure}

\subsection{\ESMT and Large Polygons with Large Aspect Ratios} \label{final_proofs}

In this section, we consider the \ESMT problem when the terminal set is formed by the vertices of $2$-CPR $n$-gons, namely $\{A_i\}$ and $\{B_i\}$. As mentioned earlier, $\{A_i\}$ is the inner polygon and $\{B_i\}$ is the outer polygon of this set of $2$-CPR $n$-gons. In particular, we consider the case when $n \geq 13$; for $n \leq 12$ these are constant sized input instances and can be solved using any brute-force technique. We also require that the aspect ratio 
$\lambda$ has a lower bound $\lambda_1$, i.e. we do not want the two polygons to have sides of very similar length. The exact value of  $\lambda_1$ will be clear during the description of the algorithm. Intuitively, when $\lambda$ is \emph{very large}, the SMT should look similar to what was derived in~\cite{weng1995steiner}. In other words, (please refer to Figure \ref{sing_con_top_fig}):
\begin{enumerate}
    \item for some $j \in [n]$, there is a  vertical fork connecting the two consecutive inner polygon points $A_j, A_{j + 1}$  with the two consecutive outer polygon points $B_j, B_{j + 1}$ - we refer to this vertical fork as the \emph{vertical gadget} for the SMT,
    \item the other points in $\{B_i\}$ are connected directly via $(n - 2)$ outer polygon edges,
    \item the other points in $\{A_i\}$ are connected via $(n - 2)$ inner polygon edges. 
\end{enumerate}

We call such a topology, a \emph{singly connected topology} as in Figure \ref{sing_con_top_fig}. 
For the rest of this section, we consider the SMTs for a large enough aspect ratio, $\lambda$ and show that there is an SMT that must be a realisation of a \emph{singly connected topology}. We refer to an SMT for the terminal set defined by the vertices of $\{A_i\}$ and $\{B_i\}$ as the \smtpoly. 

Without loss of generality, we consider the edge length of any side $A_iA_{i+1}$ in $\{A_i\}$ to be $1$. As we defined the aspect ratio to be $\lambda$, any side $B_iB_{i+1}$ of $\{B_i\}$ must have a side length of $\lambda$. Further, we observe that for any SMT $\mathcal T$, specifying $E(\mathcal T)$ sufficiently determines the entire tree, as $V(\mathcal T) = \{P~|~\exists Q \text{ such that } PQ \in E(\mathcal T)\}$.

We start with the following formal definitions:

\begin{definition}\label{def:singly-conn}
A Steiner topology of $\{A_i\} \cup \{B_i\}$ is a \textbf{singly connected topology}, if it has the following structure:

\begin{itemize}
    \item A vertical gadget i.e.  five edges $\{A_jS_a, A_{j + 1}S_a, S_aS_b, S_bB_j, S_bB_{j + 1}\}$ for some $1 \le j \le n$, where $S_a$ and $S_b$ are newly introduced Steiner points contained in the isosceles trapezoid $\{A_j, A_{j + 1}, B_j, B_{j + 1}\}$. 
    \item All $(n - 2)$ polygon edges of $\{A_i\}$ excluding the edge $A_jA_{j + 1}$
    \item All $(n - 2)$ polygon edges of $\{B_i\}$ excluding the edge $B_jB_{j + 1}$
\end{itemize} 
\end{definition}

\begin{figure}[h]
\centering
\subfloat[\centering SMT of $2$-CPR $13$-gons]{\includegraphics[width=4cm]{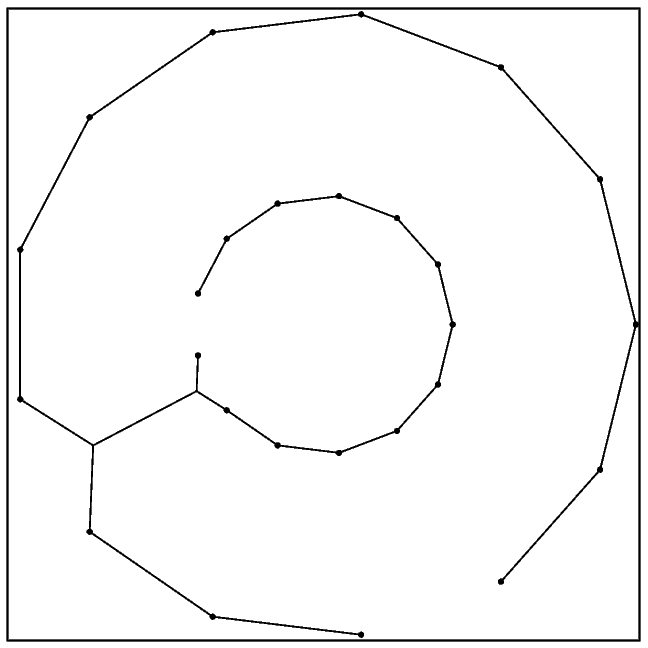}}
    \qquad
\subfloat[\centering SMT of $2$-CPR $20$-gons] {\includegraphics[width=4cm]{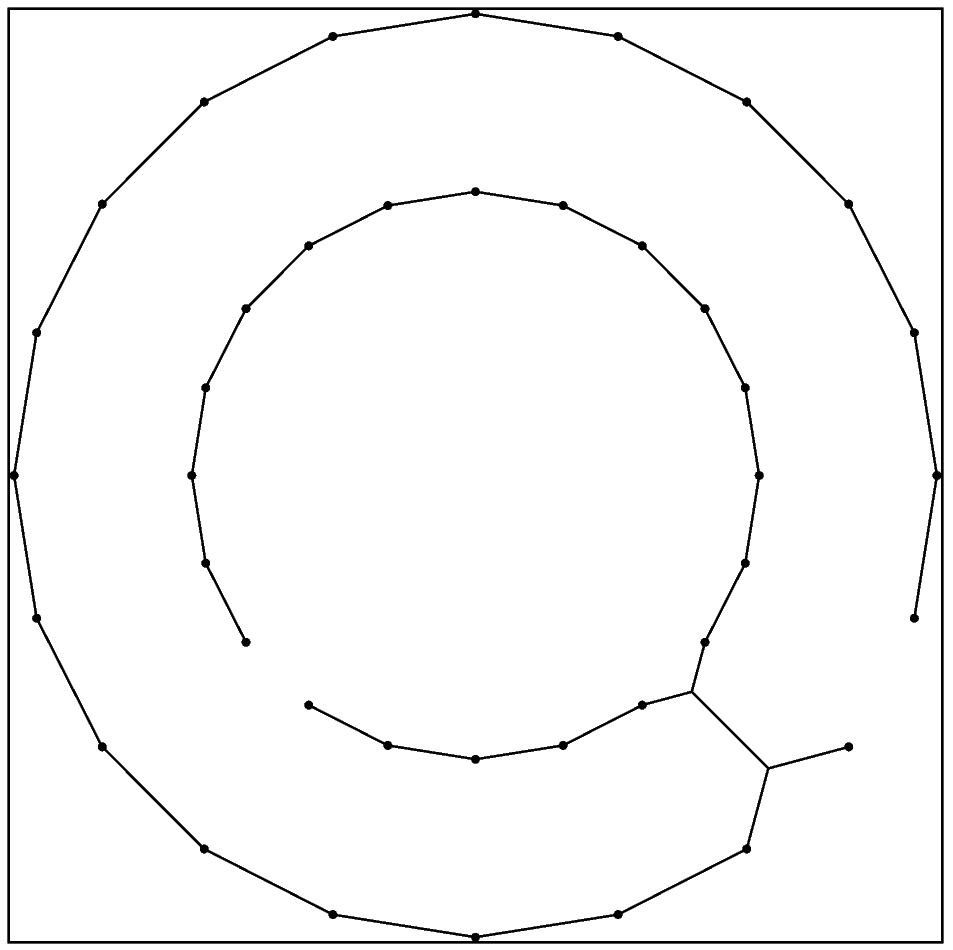} }
\caption{ \emph{Singly connected topology} of $2$-CPR $n$-gons $n = 13$ and $n = 20$ }
\label{sing_con_top_fig}
\end{figure}

We define the notion of a path in an SMT for the vertices of $\{A_i\}$ and $\{B_i\}$ where the starting point is in $\{A_i\}$ and the ending point is in $\{B_i\}$.
\begin{definition}
    An \textbf{A-B path} is a path in a Steiner tree of $\{A_i\} \cup \{B_i\}$ which starts from a vertex in $\{A_i\}$ and ends at a vertex in $\{B_i\}$ with all intermediate nodes (if any) being Steiner points. 
\end{definition}

The following Definition and Figure~\ref{fig:counterpath} is useful for the design of our algorithm.
\begin{definition}\label{def:cw_ccw_path}
    A \textbf{counter-clockwise path} is a path $P_1, P_2, ... P_m$ in a Steiner tree such that for all $i \in \{2, \ldots, m - 1\}, \angle P_{i - 1}P_{i}P_{i + 1} = \frac{2 \pi}{3}$ in the counter-clockwise direction. Similarly, a \textbf{clockwise path} is a path $P_1, P_2, ... P_m$ in a Steiner tree such that for all $i \in \{2, \ldots, m - 1\}, \angle P_{i - 1}P_{i}P_{i + 1} = \frac{4 \pi}{3}$ in the counter-clockwise direction.
\end{definition}

\begin{figure}[h]
\centering
\includegraphics[width=7cm]{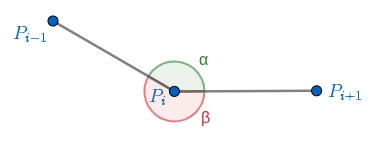}
\caption{$\angle P_{i-1}P_iP_{i+1} = \alpha$ in the counter-clockwise direction and $\angle P_{i-1}P_iP_{i+1} = \beta$ in the clockwise direction, as in~\Cref{def:cw_ccw_path}}
\label{fig:counterpath}
\end{figure}

Now, we consider any Steiner point $S$ in any SMT. Let $P$ and $Q$ be two neighbours of $S$. We now prove that there is no point of the SMT inside the triangle $PSQ$.

\begin{observation} \label{no_pt_in_PSQ}
    Let $S$ be a Steiner point in any SMT for $\{A_i\} \cup \{B_i\}$, with neighbours $P$ and $Q$. Then, no point of the SMT lies inside the triangle $PSQ$.
\end{observation}

\begin{proof}
    By the \emph{lune property} (Proposition \ref{lune}), for any edge $P_1Q_1$ in an SMT, for the two circles centred at $P_1$ and at $Q_1$, respectively and both having a radius of $\overline{P_1Q_1}$, the intersection region does not contain any point of the SMT.

    \begin{figure}[h]
    \centering
    \includegraphics[width=10cm]{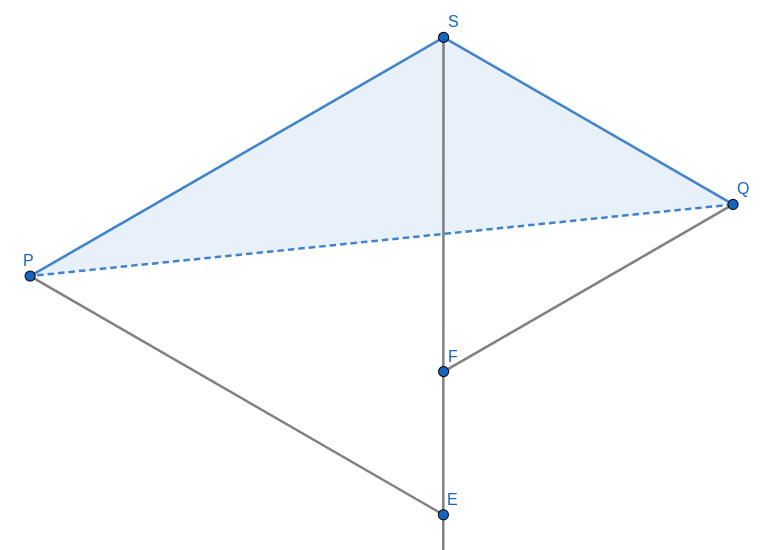}
    \caption{Observation \ref{no_pt_in_PSQ}: Given triangle $PSQ$, equilateral triangles $PSE$ and $QSF$ are constructed}
    \label{no_pt_in_PSQ_fig}
    \end{figure}

    Let $E$ and $F$ be points on the internal angle bisector of $\angle PSQ$, such that $\angle SPE = \angle SQF = \frac{\pi}{3}$ as shown in Figure \ref{no_pt_in_PSQ_fig}. Since $E$ and $F$ are points on the angle bisector of $\angle PSQ$, $\angle PSE = \angle QSF = \frac{\pi}{3}$. Hence, triangles $PSE$ and $QSF$ are equilateral triangles.

    Since $PS$ is an edge in the SMT, by the lune property, the intersection of the circles centred at $P$ and $S$, both with radius $\overline{PS}$ contain no point inside which is a part of the SMT. Since the lune contains the entire equilateral triangle $PSE$, no point of the SMT lies inside triangle $PSE$. Similarly, no point of the SMT lies inside the triangle $QSF$.

    Further, as $\angle PSQ = \frac{2 \pi}{3}$, $\angle SPQ + \angle SQP = \frac{\pi}{3}$. This means $\angle SPQ, \angle SQP < \frac{\pi}{3}$. Therefore, as $\angle SPE = \angle SQF = \frac{\pi}{3}$, $E$ and $F$ must lie outside the triangle $PSQ$. This implies that the triangle $PSQ$ is covered by the union of the triangles $PSE$ and $QSF$. As no point of the SMT lies in triangles $PSE$ and $QSF$, triangle $PSQ$ must contain no points of the SMT as well.
\end{proof}

Next, we show that in an \smtpoly there cannot be any Steiner point, in the interior of the polygon $\{A_i\}$, that is a direct neighbour of some point $B_k$ in the polygon $\{B_i\}$.

\begin{observation} \label{no_in_to_B}
    For any \smtpoly, there cannot exist a Steiner point $S$ lying in the interior of the polygon $\{A_i\}$ such that $SB_k$ is an edge in an SMT for some $B_k \in \{B_i\}$. 
\end{observation}

\begin{proof}
    For the sake of contradiction, we assume that for some \smtpoly there exists a Steiner point $S$ lying in the interior of the polygon $\{A_i\}$ such that $SB_k$ is an edge in the SMT for some $B_k \in \{B_i\}$. Let $A_mA_{m+1}$ be the edge such that $SB_k$ intersects $A_mA_{m+1}$. Without loss of generality, assume that $A_m$ is closer to $B_k$ than $A_{m+1}$. Therefore $\angle B_kA_mS > \angle B_kA_mA_{m+1} \ge \frac{\pi}{2} + \frac{\pi}{n} > \frac{\pi}{2}$. This means that $B_kS$ is the longest edge in the triangle $B_kSA_m$. Therefore we can remove the edge $B_kS$ from the SMT and replace it with either $B_kA_m$ or $SA_m$ to get another tree connecting the terminal set with a shorter total length than what we started with, which is a contradiction. 
\end{proof}

We further analyze SMTs for $\{A_i\}\cup \{B_i\}$.

\begin{observation} \label{subchain-smt}
    Let $\mathcal V = \{A_j, A_{j + 1}, \ldots, A_k\}$ be the interval of consecutive vertices of $\{A_i\}$ lying between $A_j$ and $A_k$ (which includes $A_{j + 1}$) such that $A_j$ is distinct from $A_{k + 1}$. Let $U$ be any point on the line segment $A_kA_{k+1}$. Then an SMT of $\mathcal V \cup \{U\}$ is $\mathcal T$, with $E(\mathcal T) = \{A_jA_{j+1}, A_{j+1}A_{j+2}, \ldots, A_{k-1}A_k\} \cup \{A_kU\}$. 
\end{observation}
\begin{proof}
    For the sake of contradiction, assume that there exists an SMT $\mathcal T'$ of $\mathcal V \cup \{U\}$ such that $|\mathcal T'| < |\mathcal T|$.  
    
    From \cite{du1987steiner}, we know that $\mathcal T_A$, with $E(\mathcal T_A) = \{A_jA_{j+1}, A_{j+1}A_{j+2}, \ldots, A_{j-2}A_{j-1}\}$ (\textit{i.e.} all edges of polygon $\{A_i\}$ except $A_{j-1}A_j$) is an SMT of $\{A_i\}$. Since $U \in A_kA_{k+1} \in E(\mathcal T_A)$, $\mathcal T_A$ must also be an SMT of $\{A_i\} \cup \{U\}$. However, $\mathcal T_A$ can be partitioned as $\mathcal T_A = \mathcal T \uplus \mathcal T_1$, where $E(\mathcal T_1) = \{UA_{k+1}\} \cup \{A_{k+1}A_{k+2}, \ldots, A_{j-2}A_{j-1}\}$.
    However, as $\mathcal T'$ is assumed to be of shorter total length than $\mathcal T$, $\mathcal T' \cup \mathcal T_1$ is a tree, containing $\{A_i\}$ as a vertex subset, which has a shorter total length than $\mathcal T_A$, contradicting the optimality of $\mathcal T_A$.
\end{proof}

We proceed by showing that in any \smtpoly there exists at least one A-B path which is also a counter-clockwise path. Symmetrically, we also show that for any \smtpoly there exists another clockwise A-B path which consists of only clockwise turns. We can intuitively see that this is true because, if all clockwise paths starting at a vertex in $\{A_i\}$ also ended in a vertex in $\{A_i\}$, there would be enough paths to form a cycle, which is not possible in a tree.

\begin{lemma}\label{left_right_turn_path}
In any \smtpoly, there exists an A-B path which is also a clockwise path and there exists an A-B path which is also a counter-clockwise path.
\end{lemma}
\begin{proof}
    For the sake of contradiction, assume that for some \smtpoly there is no A-B path which is a counter-clockwise path. We pick an arbitrary vertex $A_{i_1} \in \{A_i\}$ such that it is connected to at least one Steiner point (say $S_{i_1}$). We consider the counter-clockwise path, $\mathcal C_1$ starting from $A_{i_1}S_{i_1}$ and ending at the first terminal point in the counter-clockwise path. By assumption, there can be no vertex of $\{B_i\}$ in $\mathcal C_1$, hence the endpoint must be a vertex in $\{A_i\}$. Let $A_{i_2}$ be the other endpoint of $\mathcal C_1$. By definition, the penultimate vertex in this counter-clockwise path must be a Steiner point, we call it $S_{i_2}$. We again consider the counter-clockwise path starting from $A_{i_2}S_{i_2}$, and similarly, let $A_{i_3}$ be the first terminal that is encountered in this path. The penultimate vertex in this counter-clockwise path must be a Steiner point, we call it $S_{i_3}$. We can repeat this procedure indefinitely to obtain $A_{i_4}, A_{i_5}, A_{i_6}, \ldots $ as there are no counter-clockwise A-B paths. However, as $\{A_i\}$ has $n$ vertices, there must be a repetition of vertices among   $A_{i_1}, A_{i_2}, A_{i_3}, \ldots, A_{i_{n+1}}$, implying the existence of a cycle in the SMT, which is a contradiction.

    This symmetrically implies that there must also be a clockwise A-B path. 
\end{proof}

Our next step is to bound the number of `connections' that connect the inner polygon $\{A_i\}$ and the outer polygon $\{B_i\}$ for a large aspect ratio, $\lambda$. As $\lambda$ increases, the area of the annular region between the two polygons increases as well. Therefore, an increase in the number of connections would lead to a longer total length of the SMT considered. Consequently, we will prove that after a certain positive constant $\lambda_1$, for $\lambda > \lambda_1$ any \smtpoly will have a single `connection' between the two polygons. Moreover, \cite{weng1995steiner} gives us an evidence that as $\lambda \rightarrow \infty$, there will indeed be a single connection connecting the outer polygon and the inner polygon for $n \ge 12$. We can formalize this notion of existence of a single `connection' with the following lemma.\\

\begin{lemma} \label{mincut_1}
    For any \smtpoly with $n \ge 13$ and $\lambda > \lambda_{1}$, the number of edges needed to be removed in order to disconnect $\{A_i\}$ and $\{B_i\}$ is 1, where 
    $$
    \lambda_{1} = \frac{1}{1 - 4 \sin \frac{\pi}{n}}
    $$
\end{lemma}

\begin{proof}
        For the sake of contradiction, assume that for some \smtpoly, there are at least two distinct edges in that SMT, which are needed to be removed in order to disconnect $\{A_i\}$ and $\{B_i\}$. We start with a claim.
        
        \begin{claim} \label{more_than_lambda}
        A counter-clockwise A-B path in any SMT of $\{A_i\} \cup \{B_i\}$ must have an edge of length greater than $\lambda$.
        \end{claim}
        
\begin{proof}
        We consider a generic setting, where $\mathcal T$ is an SMT of some set of terminal points $\mathcal P$. Let $H \in V(\mathcal T)$ be vertex of $\mathcal T$. If $\mathcal C$ be a counter-clockwise path starting from $H$ such that no edge in the counter clockwise path has a length of more than $r$, for some $r \in \mathbb{R}^{+}$. Due to Lemma 2.4 (1) of \cite{weng1995steiner},  we know that $\mathcal C$ is contained entirely in the circle centred at $H$ with radius $2r$.
        
        In our case, any vertex in $\{A_i\}$ and any vertex in $\{B_i\}$ are separated by the distance of at least $\dfrac{\lambda - 1}{2 \sin{\frac{\pi}{n}}}$. Therefore, by the above fact, the maximum edge length in a counter-clockwise A-B path of any SMT of $\{A_i\} \cup \{B_i\}$ must be at least $\dfrac{\lambda - 1}{4 \sin{\frac{\pi}{n}}}$. Moreover, we have $\lambda > \lambda_1$. Therefore $\lambda > \dfrac{1}{1 - 4 \sin{\frac{\pi}{n}}} \implies \dfrac{\lambda - 1}{4 \sin{\frac{\pi}{n}}} > \lambda$. Hence, a counter-clockwise A-B path in any SMT of $\{A_i\} \cup \{B_i\}$ must have one edge greater than $\lambda$. This proves the claim.
\end{proof}
        
        Now, for any SMT of $\{A_i\} \cup \{B_i\}$, let $\mathcal C$ be a counter-clockwise A-B path (this exists due to Lemma \ref{left_right_turn_path}). From Claim \ref{more_than_lambda}, we know that there is an edge $e$ in $\mathcal C$, with a length greater than $\lambda$. On removing the edge $e$, the SMT splits into a forest of two trees. Let the trees be $\mathcal T_x$ and $\mathcal T_y$. As we assumed that there are two edges required to disconnect $\{A_i\}$ and $\{B_i\}$, there must exist an A-B path in either $\mathcal T_x$ or $\mathcal T_y$. Without loss of generality, let $\mathcal T_x$ contain an A-B path, and hence $\mathcal T_x$ contains at least one point from $\{A_i\}$ and at least one point from $\{B_i\}$. Further, $\mathcal T_y$ must contain at least one terminal point (as it must contain all terminal points in one of the sides of the removed edge $e$). If $\mathcal T_y$ contains a point from $\{A_i\}$, then the polygon $\{A_i\}$ has vertices both from $\mathcal T_x$ and $\mathcal T_y$; otherwise, if $\mathcal T_y$ contains a point from $\{B_i\}$, then the polygon $\{B_i\}$ has vertices both from $\mathcal T_x$ and $\mathcal T_y$. 
        
        This means that either the polygon $\{A_i\}$ or the polygon $\{B_i\}$ will contain at least one node from each of $\mathcal T_x$ and $\mathcal T_y$. Further, as any given vertex must be either in $\mathcal T_x$ or in $\mathcal T_y$, either $\{A_i\}$ or $\{B_i\}$ must contain two consecutive vertices $U_i$ and $U_{i + 1}$ such that one of them is in $\mathcal T_x$ and the other is in $\mathcal T_y$. We simply connect $U_i$ and $U_{i + 1}$ by the polygon edge which is of length $1$ (if $U_i, U_{i + 1} \in \{A_i\}$) or of length $\lambda$ (if $U_i, U_{i + 1} \in \{B_i\}$), giving us back a tree $\mathcal{T}'$containing all the terminals. However we discarded an edge of length greater than $\lambda$ and added back an edge of length at most $\lambda$ in this process, which means that the total length of $\mathcal{T}'$ is strictly less than the SMT we started with. This is a contradiction.
\end{proof}


We now proceed to further investigate the connectivity of $\{A_i\}$ and $\{B_i\}$. 

\begin{lemma} \label{steiner_path_from_A_to_A}
    Consider an \smtpoly for $n \ge 13$ and  $\lambda \ge \lambda_{1}$. There must exist $j \in [n]$ and a Steiner point $S_1$, such that terminals $A_j, A_{j + 1}$ form a path $A_j$, $S_1$, $A_{j+1}$ in the SMT and each A-B path passes through $S_1$; where 
    
    $$\lambda_{1} = \frac{1}{1 - 4 \sin \frac{\pi}{n}}$$
\end{lemma}

\begin{proof}
    From Lemma \ref{left_right_turn_path}, we know that there exists one clockwise A-B path and one counterclockwise A-B path in any SMT of $\{A_i\} \cup \{B_i\}$. Let a clockwise A-B path start from $A_r$ and a counter-clockwise A-B path start from $A_l$. Further following from Lemma~\ref{mincut_1}, as there is one edge common to all A-B paths, the clockwise A-B path from $A_r$ and the counter-clockwise A-B path from $A_l$ must share a common edge $S_1S_2$. Therefore, each A-B path must pass through $S_1$ and $S_2$. Without loss of generality we assume that point $S_1$ is closer to the polygon $\{A_i\}$ than $S_2$. This means $S_1$ is either a Stiener point or a terminal vertex of $\{A_i\}$.

    \begin{claim} \label{S_1_notin_A}
        $S_1$ is not a vertex in $\{A_i\}$
    \end{claim}
\begin{proof}
    For the sake of contradiction, we assume that $S_1$
     to be a vertex in $\{A_i\}$, let $S_1 = A_k$ in some SMT $\mathcal T_0$. We disconnect the edge $S_1S_2$ from $\mathcal T_0$, which results in the formation of a forest of two trees $\mathcal T_x$ and $\mathcal T_y$ such that $S_1 = A_k \in \mathcal T_x$ and $S_2 \in \mathcal T_y$.

    $\mathcal T_x$ must contain all vertices of $\{A_i\}$ and $\mathcal T_y$ must contain all vertices of $\{B_i\}$, as there would be an A-B path in the graph otherwise (contradicting that $S_1S_2$ disconnects $\{A_i\}$ and $\{B_i\}$). We replace $\mathcal T_x$ with the SMT of $\{A_i\}$, which is also an MST (from \cite{weng1995steiner}). Since all MST's are of the same length, we choose such an MST in which $A_k$ is not a leaf node. This means $A_{k-1}A_k$ and $A_kA_{k+1}$ are edges in the chosen MST of $\{A_i\}$. We now add back the edge $A_kS_2$, resulting in a connected tree $\mathcal T_0'$ of $\{A_i\} \cup \{B_i\}$. Since we replaced the tree $\mathcal T_x$ with an SMT of $\{A_i\}$, the total length of the $\mathcal T_0'$ must not be more than the total length of the $\mathcal T_0$.

    However, we observe that $A_k$ has three neighbours in $\mathcal T_0'$, which are $A_{k+1}, A_{k-1}, S_2$. However $\angle A_{k-1}A_kA_{k+1} > \frac{2 \pi}{3}$. This means either $\angle S_2A_kA_{k+1} < \frac{2 \pi}{3}$ or $\angle A_{k-1}A_kS_2 < \frac{2 \pi}{3}$. But due to Proposition \ref{smt-prop}, this cannot form an SMT. Therefore $\mathcal T_0'$ is not optimal; and hence, $\mathcal T_0$ cannot be optimal as well. This proves  the claim.
    \end{proof}

    Therefore, $S_1$ must be a Steiner point. Let $P$ and $Q$ be the neighbours of $S_1$ other than $S_2$, such that $\angle PS_1S_2$ is a clockwise turn while $\angle QS_1S_2$ is a counter-clockwise turn. This means that the clockwise A-B path from $A_r$ passes through $P$ and the counter-clockwise A-B path from $A_l$ passes through $Q$. We prove that $P$ and $Q$ are consecutive vertices of $\{A_i\}$ in some SMT of $\{A_i\} \cup \{B_i\}$.

    \begin{claim} \label{PQ_outside_A}
         $P$ and $Q$ cannot simultaneously lie in the interior of the polygon $\{A_i\}$. 
    \end{claim}
\begin{proof}
    We assume for the sake of contradiction that both $P$ and $Q$ lie in the interior of the polygon $\{A_i\}$. On deleting the edge $S_1S_2$, the SMT of $\{A_i\} \cup \{B_i\}$ splits into two trees $\mathcal T_1$ (rooted at $S_1$) and $\mathcal T_2$ (rooted at $S_2$). Further, as all A-B paths pass through $S_1S_2$, all vertices of $\{A_i\}$ must be in $\mathcal T_1$ whereas all vertices of $\{B_i\}$ must lie in  $\mathcal T_2$. Further, $\mathcal T_1$ must be the SMT of $\{A_i\} \cup \{S_1\}$ and $\mathcal T_2$ must be the SMT of $\{B_i\} \cup \{S_2\}$.

    \begin{itemize}
    \item \textit{Case I: One point in $\{S_1, S_2\}$ lies in the interior of $\{A_i\}$ and the other point lies in the exterior of $\{A_i\}$:} This means that the edge $S_1S_2$ crosses some polygon edge of $\{A_i\}$, call it $A_mA_{m+1}$. Let $D$ be the intersection of $A_mA_{m+1}$ and $S_1S_2$. We replace $\mathcal T_1$ with an MST of $\{A_i\}$ that contains the edge $A_mA_{m+1}$ (this can never lead to increase in total tree length due to \cite{weng1995steiner}) and remove the line segment $S_1D$ from $\mathcal T_2$. This forms a tree connecting the terminal set $\{A_i\} \cup \{B_i\}$ which has a total length smaller than the SMT we started with, which is a contradiction.
    
    \item \textit{Case II: Both $S_1$ and $S_2$ lie in the interior of polygon $\{A_i\}$:} We further consider two cases for this:

    \begin{itemize}
        \item Consider that there is at least one polygon edge $A_mA_{m+1}$ of $\{A_i\}$ such that it does not intersect with $\mathcal T_2$. Then we can replace $\mathcal T_1$ by the MST of $\{A_i\}$ which does not contain the edge $A_mA_{m+1}$ without reducing the total edge. However, this will be a connecting tree of $\{A_i\} \cup \{B_i\}$ with a smaller total length than the tree we started with (as we had removed the edge $S_1S_2$ previously), which is a contradiction.
        \item Now, consider that all polygon edges of $\{A_i\}$ intersect with some edge in $\mathcal T_2$. Since $\mathcal T_2$ is rooted at $S_2$ which lies in the interior of $\{A_i\}$, there must be $n$ distinct edges crossing the polygon $\{A_i\}$. However, $\mathcal T_2$ must be the SMT of the points $\{S_2\} \cup \{B_i\}$, which means there can be at most $(n - 1)$ Steiner points other than $S_2$ (from~\Cref{smt-prop}). Therefore, one of these $n$ edges must have a point in $\{B_i\}$ as one of its endpoints, contradicting Observation \ref{no_in_to_B}.

    \end{itemize}
    
    \item \textit{Case III: Both $S_1$ and $S_2$ lie in the exterior of polygon $\{A_i\}$:} This means that the edges $S_1P$ and $S_1Q$ intersect the polygon edges of $\{A_i\}$. Further, from Observation \ref{no_pt_in_PSQ}, there cannot be any terminal inside the triangle $S_1PQ$. Hence, $S_1P$ and $S_1Q$ must intersect the same polygon edge of $\{A_i\}$ (otherwise intermediate vertices from $\{A_i\}$ would lie in the triangle $S_1PQ$). Let this edge be $A_tA_{t+1}$. Let $P_1$ and $Q_1$ be the points of intersection of $A_tA_{t+1}$ with $S_1P$ and $S_1Q$ respectively.

    We now remove the line segments $S_1P_1$ and $S_1Q_1$ from $\mathcal T_1$. This results in another split into two connected trees $\mathcal T_P$ (containing $P_1$, $P$ and a subset of $\{A_i\}$) and $\mathcal T_Q$ (containing $Q_1$, $Q$ and the remaining vertices of $\{A_i\}$). 
    
    We observe that the terminals in $\mathcal T_P$ and $\mathcal T_Q$ form consecutive intervals of the edges in $\{A_i\}$. To see why, consider the opposite, \textit{i.e.} there are vertices $A_{i_1}, A_{i_2}, A_{i_3}, A_{i_4}$ appearing in that order in $\{A_i\}$ such that $A_{i_1}, A_{i_3} \in \mathcal T_P$ whereas $A_{i_2}, A_{i_4} \in \mathcal T_Q$. As $\mathcal T_P$ and $\mathcal T_Q$ lie in the interior of $\{A_i\}$, the path from $A_{i_1}$ to $A_{i_3}$ in $\mathcal T_P$ must cross the path from $A_{i_2}$ to $A_{i_4}$ in $\mathcal T_Q$. However, there cannot be crossing paths in the original SMT of $\{A_i\} \cup \{B_i\}$ (due to~\Cref{smt-prop}). 
    
    Let $\{A_{w}, A_{{w}+1}, \ldots, A_t\}$ be the terminals in $\mathcal T_P$ and the remaining terminals in $\{A_i\}$ are in $\mathcal T_Q$. Again, let $\mathcal T_P'$, $\mathcal T_Q'$ be defined as:
    $$E(\mathcal T_P') = \{A_wA_{w+1}, A_{w+1}A_{w+2}, \ldots, A_{t-1}A_{t}\} \cup \{A_tP_1\}$$
    and
    $$E(\mathcal T_Q') = \{A_{t+1}A_{t+2}, \ldots, A_{w-2}A_{w-1}\} \cup \{Q_1A_{t-1}\}$$
    From Observation \ref{subchain-smt}, we know that $\mathcal T_P'$ is the SMT of $\{A_{w}, A_{{w}+1}, \ldots, A_t\} \cup \{P_1\}$ and $\mathcal T_Q'$ is the SMT of $\{A_{t + 1}, A_{{t}+2}, \ldots, A_{w-1}\} \cup \{Q_1\}$. Therefore, $|\mathcal T_P'| \le |\mathcal T_P|$ and $|\mathcal T_Q'| \le |\mathcal T_Q|$. This means that $| \mathcal T_1 | \ge |\mathcal T_1'|$, where $\mathcal T_1' = \{S_1P_1, S_1Q_1\} \cup \mathcal T_P' \cup \mathcal T_Q'$. Further, $\mathcal T_1'$ is also a connecting tree of $\{S_1\} \cup \{A_i\}$ and as $\mathcal T_1$ is an SMT of $\{S_1\} \cup \{A_i\}$, then $\mathcal T_1'$ must also be an SMT with $|
    \mathcal T_1'| = |\mathcal T_1|$.

    However, We can remove $S_1P_1$ and $P_1A_t$ from $\mathcal T_1'$ and add $S_1A_t$ to get another connecting tree of $\{S_1\} \cup \{A_i\}$, but with shorter total length (as $\overline{S_1P_1} + \overline{P_1A_t} > \overline{S_1A_t}$ from triangle inequality). This contradicts the optimality of $\mathcal T_1'$ which was derived to be an SMT of $\{S_1\} \cup \{A_i\}$.

    \end{itemize}
This proves the claim.
\end{proof}

    We proceed to prove a stronger claim regarding $P$ and $Q$.
    
    \begin{claim} \label{PQ_cons}
    $P$ and $Q$ are consecutive vertices of $\{A_i\}$ in any SMT of $\{A_i\} \cup \{B_i\}$.
    
    \end{claim}
    \begin{proof}
    We first prove that $P, Q$ are vertices of $\{A_i\}$. 
    
    From Claim \ref{PQ_outside_A}, we know that at least one among $P$ and $Q$ must not be in the interior of polygon $\{A_i\}$. Without loss of generality, let it be $P$. We now show that $P$ is a vertex of $\{A_i\}$. For the sake of contradiction we assume that $P$ is not a vertex of $\{A_i\}$ \textit{i.e.} $P$ is a Steiner point. Let $\overrightarrow{PF_1}$ and $\overrightarrow{PF_2}$ be tangents from $P$ to $\{A_i\}$ where $F_1, F_2$ are the points of tangency on $\{A_i\}$. As $\overrightarrow{PF_1}$ and $\overrightarrow{PF_2}$ are tangents, $\angle F_1PF_2 < \pi$. We denote the region between the tangents $\overrightarrow{PF_1}$ and $\overrightarrow{PF_2}$ which contains the all the points in $\{A_i\}$ as $\mathcal R$.
    
    From any Steiner point $H$, which lies outside $\mathcal R$, we can choose a neighbour $H_1$ of $H$ such that $\overrightarrow{HH_1}$ is not directed towards $\mathcal R$. Further we now show that there is one neighbour $P_1$ of $P$ such that $P_1$ is not in $\mathcal R$ and $P_1 \neq S_1$. 

    \textit{Case I: $S_1$ lies in $\mathcal R$.} However there must be another neighbour $P_1$ of $P$ not in $\mathcal R$ (as $\angle F_1PF_2 < \pi$) but as $P_1$ is outside $\mathcal R$, we must have $P_1 \neq S_1$.
    
    \textit{Case II: $S_1$ does not lie in $\mathcal R$ (Figure \ref{steiner_path_from_A_to_A_fig}).} As the counter-clockwise A-B path from $A_l$ passes through $Q$, $A_l$ must be to the left of the line $L_{QS_1}$ if the line is given a orientation from $Q$ to $S_1$. This means that one of the tangents from $P$ (Without loss of generality assume it to be $\overrightarrow{PF_1}$) intersects with the line $L_{QS_1}$. Therefore, taking angles in counter-clockwise order, we have:
    
\begin{figure}[h]
\centering
\includegraphics[width=13cm]{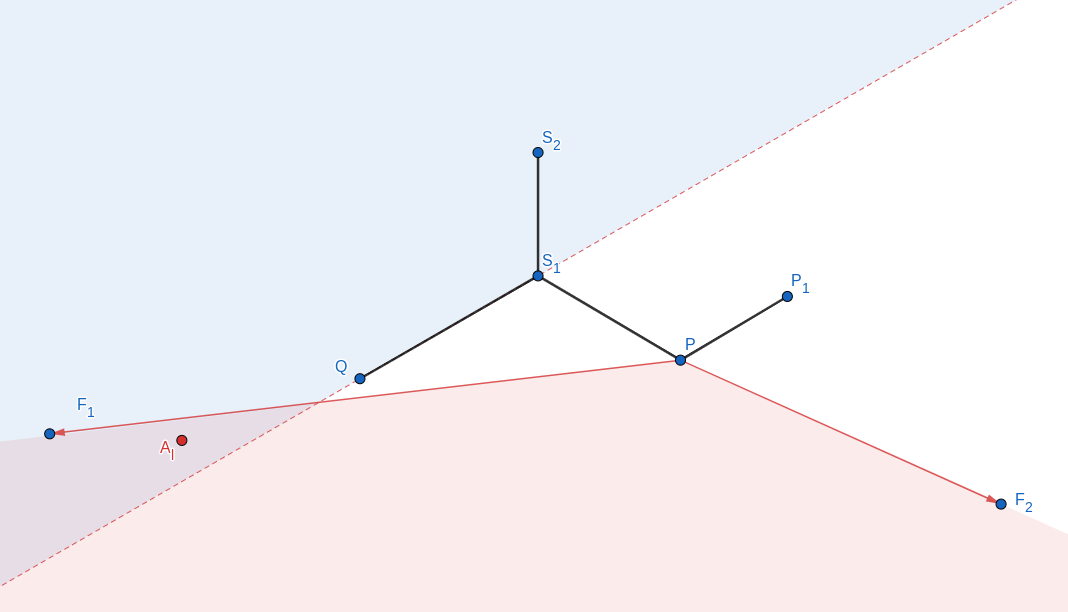}
\caption{Case II of~\Cref{PQ_cons}}
\label{steiner_path_from_A_to_A_fig}
\end{figure}

    \begin{align*}
        & \angle S_1PF_1 < \frac{\pi}{3} & [\text{as } \overrightarrow{PF_1} \text{ intersects } L_{QS_1}] \\
        \implies & \angle S_1PF_2 = \angle S_1PF_1 + \angle F_1PF_2 < \frac{\pi}{3} + \pi = \frac{4 \pi}{3}\\
        \implies & \angle F_2PS_1 = 2 \pi - \angle S_1PF_2 > 2 \pi - \frac{4 \pi}{3} = \frac{2 \pi}{3}
    \end{align*}
    
    Hence there must exist one neighbour $P_1$ of $P$ lying outside the $\mathcal R$, precisely in the region bounded by the rays $\overrightarrow{PF_2}$ and $\overrightarrow{PS_1}$ with $P_1 \ne S_1$. 
    
    Further we can choose a neighbour $P_2$ of $P_1$ such that $\overrightarrow{P_1P_2}$ is directed away from $\mathcal R$. We can continue choosing $P_2, P_3, \ldots $ such that $\overrightarrow{P_iP_{i+1}}$ is directed away from the region $\mathcal R$. Moreover, the path  $P, P_1, P_2, \ldots$ must end at some point $B_k$ as it cannot end in any vertex of $\{A_i\}$ (since all vertices of $\{A_i\}$ are in $\mathcal R$). Now, let $\mathcal C_1$ be the path from $A_r$ to $B_k$ (which passes through $P$ and $P_1$) and let $\mathcal C_2$ be the counter-clockwise A-B path from $A_l$ (passing through $Q$, $S_1$ and $S_2$). We observe that $\mathcal C_1$ and $\mathcal C_2$ are two edge disjoint A-B paths, which is a contradiction to~\Cref{mincut_1}. This proves that $P$ is indeed a vertex in $\{A_i\}$. Therefore by Claim \ref{PQ_outside_A}, $Q$ does not lie inside $\{A_i\}$ and repeating this same argument on $Q$ yields that $Q$ is also a vertex of $\{A_i\}$. 

    Now, to prove that $P$ and $Q$ are consecutive vertices of $\{A_i\}$, we use Observation \ref{no_pt_in_PSQ}. Observation \ref{no_pt_in_PSQ} implies that there must not be any other point of the SMT in the triangle $PS_1Q$. This means that $P$ and $Q$ must be consecutive vertices of $\{A_i\}$, otherwise all polygon vertices of $\{A_i\}$ occurring in between $P$ and $Q$ would be inside the triangle $PS_1Q$ (as $\angle PS_1Q = \dfrac {2 \pi}{3}$ and $n \ge 13$). This proves the claim.
    \end{proof}

    Therefore, $P$ and $Q$ are consecutive vertices $A_j, A_{j+1}$ of the polygon $\{A_i\}$, for some $j \in [n]$ such that $A_j$, $S_1$, $A_{j+1}$ is a path in the SMT, where $S_1$ is a Steiner point lying on all A-B paths.
\end{proof}

Our next step is to investigate some more structural properties of an SMT for $\{A_i\} \cup \{B_i\}$. From \cite{du1987steiner}, we may guess that there would be a lot of polygon edges of both $\{A_i\}$ and $\{B_i\}$ in an SMT. We prove the following Lemma, stating that there is an SMT of $\{A_i\} \cup \{B_i\}$ which contains $(n - 2)$ polygon edges of $\{A_i\}$.

\begin{lemma} \label{n-2_A_poly_edges}
    For an \smtpoly with aspect ratio $\lambda$, 
    $\lambda > \lambda_{1} = \frac{1}{1 - 4 \sin \frac{\pi}{n}}$, let $S_1$ be the Steiner point such that all A-B paths pass through $S_1$. Let $A_j$ and $A_{j+1}$ be vertices of $\{A_i\}$ which are connected to $S_1$. Then, there exists an SMT of $\{A_i\} \cup \{B_i\}$ having $(n - 2)$ polygon edges of $\{A_i\}$ other than $A_jA_{j+1}$.
\end{lemma}

\begin{proof}
    Let $\mathcal T_0$ be any SMT of $\{A_i\} \cup \{B_i\}$. From Lemma \ref{steiner_path_from_A_to_A}, we know that there exists $S_1$, $A_j$ and $A_{j + 1}$ such that $S_1$ is a Steiner point which is a part of all A-B paths, and $A_jS_1A_{j+1}$ is a path $\mathcal T_0$.

    From $\mathcal T_0$, we remove the edges $S_1A_j$ and $S_1A_{j + 1}$ and add the edge $A_jA_{j+1}$. This results in a forest of two disjoint trees $\mathcal T_x$ and $\mathcal T_y$. One of these trees (say $\mathcal T_x$) must contain all terminal points from $\{A_i\}$ and the other tree must contain all terminals from $\{B_i\}$, as no more A-B paths exist after we removed the edges $S_1A_j$ and $S_1A_{j+1}$. Therefore we have $|\mathcal T_0| = |\mathcal T_x| + |\mathcal T_y| - 1 + \overline{S_1A_j} + \overline{S_1A_{j+1}}$.

    We further replace $\mathcal T_x$ with a Euclidean minimum spanning tree $\mathcal T_x'$ of $\{A_i\}$ such that the edge $A_jA_{j+1}$ is present in $\mathcal T_x'$. From \cite{du1987steiner}, we know that $|\mathcal T_x'| \le |\mathcal T_x|$. We now remove the edge $A_jA_{j+1}$ and add back the edges $S_1A_j$ and $S_1A_{j+1}$ which gives a connected tree $\mathcal T_0'$ of $\{A_i\} \cup \{B_i\}$. Therefore we have:
    $$|\mathcal T_0'| = |\mathcal T_x'| + |\mathcal T_y| - 1 + \overline{S_1A_j} + \overline{S_1A_{j+1}} \le |\mathcal T_x| + |\mathcal T_y| - 1 + \overline{S_1A_j} + \overline{S_1A_{j+1}} = |\mathcal T_0|$$

    This means $\mathcal T_0'$ must be an SMT. However, all polygon edges of polygon $\{A_i\}$ appearing in $\mathcal T_x'$ also appear in $\mathcal T_0'$ as well, except $A_jA_{j+1}$. Therefore, the SMT $\mathcal T_0'$ has $(n - 2)$ polygon edges of the polygon $\{A_i\}$.
\end{proof}

With these set of results in hand, we can now show that there exists an SMT of $\{A_i\} \cup \{B_i\}$ following a \emph{singly connected topology}. To show this, we start with any \smtpoly, $\mathcal T_0$, that satisfies all the results derived so far and transform it into a Steiner tree of \emph{singly connected topology} having total length not longer than the initial Steiner tree $\mathcal T_0$.

\begin{theorem} \label{final_proof}
There exists an \smtpoly following a singly connected topology for $n \ge 13$ and $\lambda \ge \lambda_{1}$, where
$$
\lambda_{1} = \frac{1}{1 - 4 \sin \frac{\pi}{n}}
$$
\end{theorem}

\begin{proof}

Let $\mathcal{T}_0$ be any SMT of $\{A_i\} \cup \{B_i\}$ which satisfies the properties of Lemma \ref{n-2_A_poly_edges}. Further, from Lemma \ref{steiner_path_from_A_to_A}, there is a Steiner point $S_1$ which lies on all A-B paths, and there are two consecutive vertices $A_j$, $A_{j+1}$ such that $A_j$, $S_1$, $A_{j+1}$ is a path in $\mathcal T_0$. As $\mathcal T_0$ satisfies the property of Lemma \ref{n-2_A_poly_edges}, $\mathcal T_0$ has $(n - 2)$ polygon edges of $\{A_i\}$ excluding the edge $A_jA_{j+1}$.

Let $H$ be the point in the interior of the polygon $\{A_i\}$  such that $HA_jA_{j+1}$ form an equilateral triangle. As $n > 6$, the common centre $O$ of $\{A_i\}$ and $\{B_i\}$ does not lie inside the triangle $HA_jA_{j+1}$. Now, we modify $\mathcal T_0$ as follows:

\begin{enumerate}
    \item Remove edges $A_jS_1$, $S_1A_{j+1}$ and add edge $S_1H$ to get the forest $\mathcal{T}_1$. We know from \cite{hwang1992steiner} that $S_1$, $S_2$ and $H$ are collinear and this transformation does not change the total length. Therefore $|\mathcal{T}_0| = |\mathcal{T}_1|$. Here, $|\mathcal{T}_1|$ denotes the sum of the lengths of edges present in $\mathcal{T}_1$.
    \item Add edge $HO$ and remove all polygon edges of $\{A_i\}$ to get $\mathcal T_2$. Therefore $|\mathcal T_2| = |\mathcal T_1| + \overline{HO} - (n - 2) = |\mathcal T_0| + \overline{HO} - (n - 2)$. We observe that $\mathcal T_2$ is a tree connecting  the points in $\{B_i\} \cup \{O\}$.
    \item Let $S_0$ be the Torricelli point of the triangle $OB_jB_{j+1}$. Let $\mathcal T_3$ be the Steiner tree of $\{B_i\} \cup \{O\}$ with edges $S_0O$, $S_0B_j$, $S_0B_{j+1}$ and other points in $\{B_i\}$ connected through $(n - 2)$ polygon edges of the polygon $\{B_i\}$. From \cite{weng1995steiner}, we know that $\mathcal T_3$ is the SMT of $\{B_i\} \cup \{O\}$. Therefore $|\mathcal T_3| \le |\mathcal T_2| = |\mathcal T_0| + \overline{HO} - (n - 2)$. Further we know that $H$ lies on the edge $OS_0$ (as $O$, $S_0$ and $H$ lie on the perpendicular bisector of $B_j$ and $B_{j + 1}$).
    \item Remove edge $S_0O$ and add edge $S_0H$ to get $\mathcal T_4$. As $H$ lies on the edge $OS_0$, we have $|\mathcal T_4| = |\mathcal T_3| - \overline{OH} \le |\mathcal T_0| - (n - 2)$.
    \item Let $S_3$ be the intersection of the circumcircle of triangle $A_jHA_{j + 1}$ (from Lemma \ref{lambda_v} the intersection exists as $\lambda_1 \ge \lambda_v$ for $n \ge 13$). Remove the edge $S_3H$ and add the edges $S_3A_j$ and $S_3A_{j + 1}$ to get $\mathcal T_5$. Again, from \cite{hwang1992steiner} we know that this transformation does not change the total length. Hence $|\mathcal T_5| = |\mathcal T_5| \le |\mathcal T_0| - (n - 2)$. Moreover, as $\lambda > \lambda_v$, we observe that $\{A_j, B_j, A_{j+1}, B_{j+1}, S_3, S_0\}$ form the vertices of the vertical gadget and points $O$, $H$, $S_3$, $S_0$ appear in that order on the perpendicular bisector of $B_j$ and $B_{j+1}$.
    \item Add back the $(n - 2)$ polygon edges of $\{A_i\}$ which were removed in the second step to get $\mathcal T_6$. Therefore $|\mathcal T_4| = |\mathcal T_5| + (n - 2) \le |\mathcal T_0|$. We further observe that $\mathcal T_6$ is a Steiner tree connecting the points $\{A_i\} \cup \{B_i\}$ with a singly connected topology.
\end{enumerate}

Therefore we started with an arbitrary SMT $\mathcal T_0$ and transformed it into a Steiner tree $\mathcal T_6$ with a singly connected topology (where $\{A_j, B_j, A_{j+1}, B_{j+1}, S_3, S_0\}$ form the vertices of the vertical gadget) which has a total length not worse than $\mathcal T_0$. Hence $\mathcal T_6$ must be an SMT of $\{A_i\} \cup \{B_i\}$. This proves the theorem.
\end{proof}

\begin{remark}
    \cref{final_proof} determines the exact structure of the \smtpoly. Further from~\cref{trapezoids} we determine the exact method to construct the two additional Steiner points in $\mathcal O(1)$ steps - note that this construction time is independent of the integer $n$ or the real number $\lambda$. Therefore, \smtpoly for $n \ge 13$ and $\lambda \ge \lambda_1$ is solvable in polynomial time.
\end{remark}


Note that the total length of any \smtpoly, when $n \ge 13$ and $\lambda \ge \lambda_{1}$, is
\begin{align*}
    & |\mathcal T_6| = |\text{vertical gadget}| + |(n - 2) \text{ edges of $\{B_i\}$}| + |(n - 2) \text{ edges of $\{A_i\}$}| \\
    \implies & |\mathcal T_6| = \bigg(\dfrac{(\lambda - 1)}{2 \tan \frac \pi n} + \dfrac {\sqrt 3 (\lambda + 1)} {2}\bigg) + (n - 2) \cdot \lambda + (n - 2)\\
    \implies & {|\mathcal T_6| = \dfrac{(\lambda - 1)}{2 \tan \frac \pi n} + \bigg(n - 2 + \frac{\sqrt{3}}{2}\bigg) (\lambda + 1)}\\
\end{align*}
\noindent Further, $\lambda_1$ converges to 1 very quickly with increasing $n$ (plotted in Figure \ref{lambda_1_plot}):

    \vspace{0.3cm}
    \begin{center}
    \begin{tabular}{| c | c | c | c | c | c |}
    \hline
    $n$ & 13 & 20 & 40 & 100 & 500 \\
    \hline
    $\lambda_1$ & 23.3987 & 2.6719 & 1.4574 & 1.1437 & 1.0258\\
    \hline
    \end{tabular}
    \end{center}
    \vspace{0.3cm}

\begin{figure}[h]
\centering
\includegraphics[width=15cm]{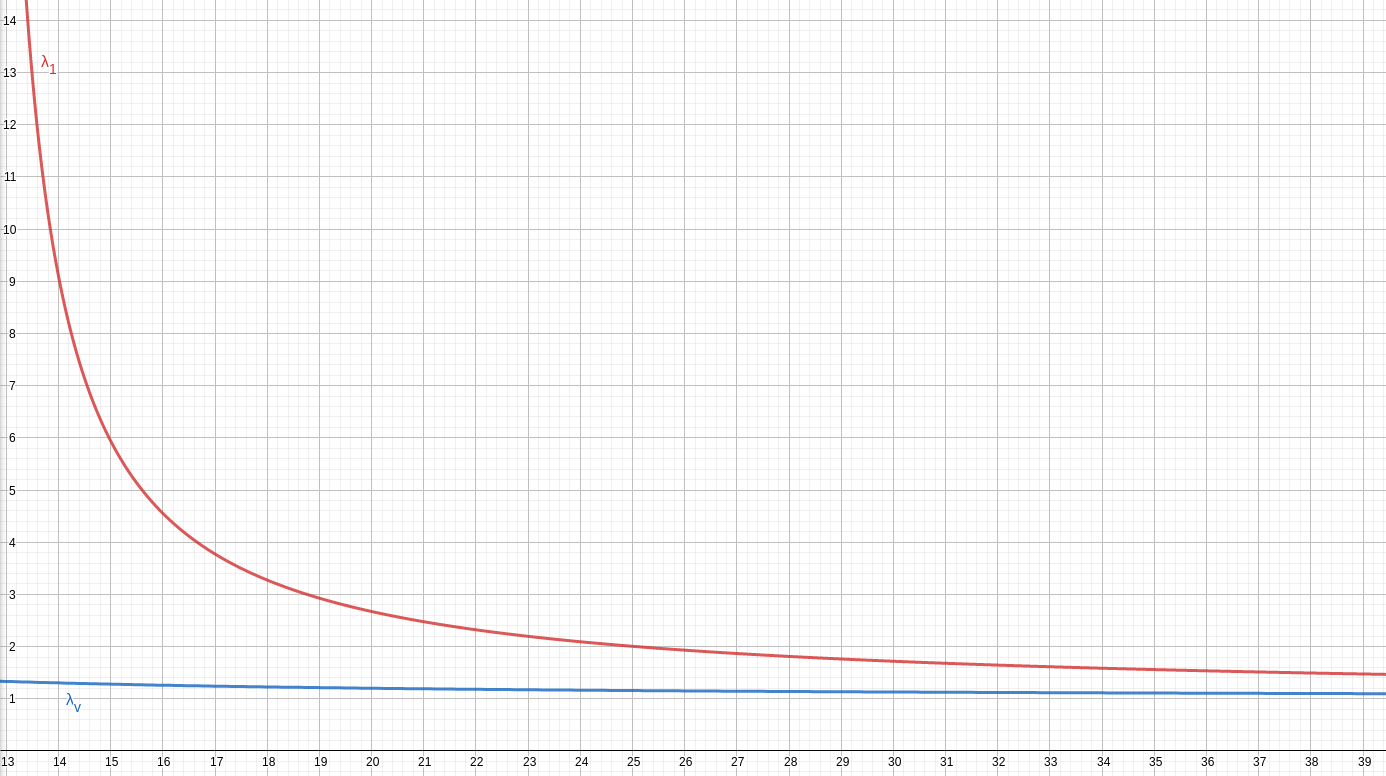}
\caption{Plot of $\lambda_1$ \& $\lambda_v$ against $n$}
\label{lambda_1_plot}
\end{figure}

This means for large sized $n$ and for ratios that are not too small, the SMT will follow a \textit{singly connected topology}.

\section{\ESMT on $f(n)$-Almost Convex Point Sets}\label{sec:exact_algo}


In this section, we design an exact algorithm for \ESMT on $f(n)$-Almost Convex Point Sets running in time $2^{\OO(f(n)\log n)}$. Note that $f(n) \leq n$ is always true. Therefore, we are given as input a set $\mathcal{P}$ of $n$ points in the Euclidean Plane such that $\mathcal P$ can be partitioned as $\mathcal P = \mathcal P_1 \uplus \mathcal P_2$, where $\mathcal P_1$ is the convex hull of $\mathcal P$ and $|\mathcal P_2| = f(n)$.

First, we look into some mathematical results and computational results to finally arrive at the algorithm for solving \ESMT on $f(n)$-Almost Convex Point Sets. 

We know that the SMT of $\mathcal P$ can be decomposed uniquely into one or more full Steiner subtrees, such that two full Steiner subtrees share at most one node~\cite{hwang1992steiner}. In the following lemma, we further characterize one full Steiner subtree.

\begin{lemma}
\label{existence_of_leaf}
Let $\mathbb{F}$ be the full Steiner decomposition of an SMT of $\mathcal{P}$.  Then there exists a full Steiner subtree $\mathcal F \in \mathbb{F}$ such that $\mathcal F$ has at most one common node with at most one other full Steiner subtree in $\mathbb{F}$.
\end{lemma}
\begin{proof}
If the SMT of $\mathcal P$ is a full Steiner tree, then the statement is trivially true.

Otherwise, we assume that the SMT of $\mathcal P$ has full Steiner subtrees, $\mathbb{F} = \{\mathcal F_1, \mathcal F_2, \ldots, \mathcal F_m\}$, $m \ge 2$. Now, for the sake of contradiction, we assume that for each full Steiner subtree $\mathcal F_j$ there are atleast two other full Steiner subtrees ${\sf Neitree}(\mathcal F_j)^1,{\sf Neitree}(\mathcal F_j)^2 \in \mathbb{F}$ and two terminals $P_1(\mathcal F_j),P_2(\mathcal F_j) \in \mathcal{P}$ such that $P_i(\mathcal F_j) \in V(\mathcal F_j) \cap V({\sf Neitree}(\mathcal F_j)^i)$, $i \in \{1,2\}$. Now, let us construct a walk $W$ in the SMT of $\mathcal{P}$. Starting from $P_1(\mathcal F_1)$ of the full Steiner subtree $\mathcal F_1 \in \mathbb{F}$, we include the path in $\mathcal F_1$ connecting to $P_2(\mathcal F_1)$. Note that $P_2(\mathcal F_1)$ is also contained in ${\sf Neitree}(\mathcal F_1)^2$. Let ${\sf Neitree}(\mathcal F_1)^2 = \mathcal F_{w_1}$ for some $w_1\in [m], w_1\neq 1$. Also let $P_2(\mathcal F_1) = P_1(\mathcal F_{w_1})$.
Then, we know that there is a $P_2(\mathcal F_{w_1})$. In $W$, we include the path in $\mathcal F_{w_1}$ connecting $P_1(\mathcal F_{w_1})$ to $P_2(\mathcal F_{w_1})$. In general, suppose the $i^{th}$ full Steiner subtree to be considered in building the walk is $\mathcal F_{w_{i-1}}$ which was reached via point $P_1(\mathcal F_{w_{i-1}})$. Then we include in $W$ the path in $\mathcal F_{w_{i-1}}$ connecting $P_1(\mathcal F_{w_{i-1}})$ and $P_2(\mathcal F_{w_{i-1}})$. Thus, we can indefinitely keep constructing the walk $W$ as for each $\mathcal F_{w_{i-1}}$ both $P_1(\mathcal F_{w_{i-1}}), P_2(\mathcal F_{w_{i-1}})$ always exist. However, since there are $m$ full Steiner subtrees this means that there is an $\mathcal F_k \in \mathbb{F}$ and two indices $i\neq j$ such that $\mathcal F_k = \mathcal F_{w_i} = \mathcal F_{w_j}$. Thus, there exists a cycle in $W$, which implies that there is a cycle in the SMT of $\mathcal{P}$ (contradiction). Therefore, there must be at least one full Steiner subtree that has at most one common terminal with at most one other full Steiner subtree.
\end{proof}
\begin{remark}
A full Steiner subtree of the SMT of $\mathcal P$ has the topology of a tree. Thus, from Lemma~\ref{existence_of_leaf}, we conclude that a full Steiner subtree, that has at most one common terminal with at most one other full Steiner subtree, has at least one leaf of the SMT.
\end{remark}

\begin{definition}
Let the full Steiner subtrees, that have at most one terminal shared with at most one other full Steiner subtree, be called \textbf{leaf full Steiner subtrees}. Let the terminal which is shared be called the \textbf{pivot} of the leaf full Steiner subtree.
\end{definition}

\begin{figure}[h]
\centering
\subfloat[\centering Full Steiner Subtrees of Figure 1(a)] {\includegraphics[width=6cm]{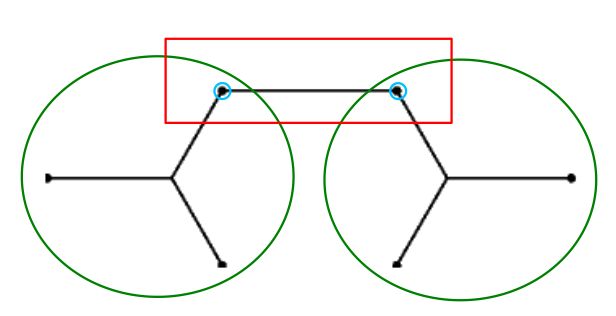}}
    \qquad
\subfloat[\centering \centering Full Steiner Subtrees of Figure 1(b)] {\includegraphics[width=6cm]{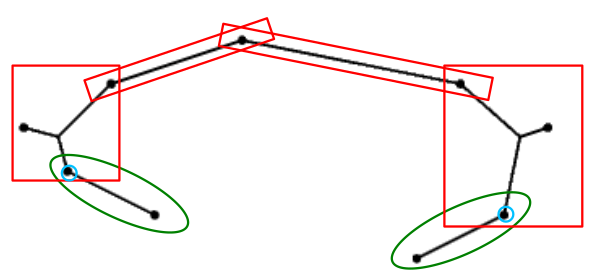}}
\caption{Leaf full Steiner subtrees enclosed in ellipses, other full Steiner subtrees enclosed in rectangles, pivots of leaf full Steiner subtrees encircled}
\end{figure}

\begin{lemma}
\label{leaf_non_leaf_structure}
Let $\mathcal F$ be a leaf full Steiner Subtree of the SMT of $\mathcal P$, with terminal points $\mathcal P_{\mathcal{\mathcal F}} \subseteq \mathcal{P}$ and having pivot $P_{\mathcal F}$. Deleting $\mathcal F\setminus \{P_{\mathcal F}\}$ from the SMT of $\mathcal P$ gives us an SMT of the terminal points $((\mathcal P - \mathcal{P}_{\mathcal F}) \cup \{P_{\mathcal F}\})$.
\end{lemma}

\begin{proof}
Firstly, we observe that deleting $\mathcal F\setminus \{P_{\mathcal F}\}$ from the SMT of $\mathcal P$ will indeed give us a tree, as $\mathcal F$ is a leaf full Steiner subtree. Let us call this tree $\mathcal Y$.

Now for the sake of contradiction, we assume that the total length of $\mathcal Y$ is strictly larger than the SMT $\mathcal F'$ of $((\mathcal P - \mathcal{P}_{\mathcal F}) \cup \{P_{\mathcal F}\})$. However, this means, the total length of $\mathcal F' \cup \mathcal F$ is strictly smaller than that of the SMT of $\mathcal P$. As $\mathcal F' \cup \mathcal F$ is also a Steiner tree of $\mathcal P$,  this contradicts the minimality of the initial SMT of $\mathcal P$.
\end{proof}

Now we are ready to describe the algorithm. Recall that $\mathcal P$ is partitioned as $\mathcal P = \mathcal{P}_1 \uplus \mathcal{P}_2$, where $\mathcal{P}_1$ is the convex hull of $\mathcal P$ and $\mathcal{P}_2$ is the set of $f(n)$ points lying in the interior of $\mathcal P_1$. For the sake of brevity of notations let $|\mathcal{P}_2| = k$.

\begin{lemma}
\label{four_power_n}
Let $\mathcal P$ be a $k$-Almost Convex Point Set. A minimum FST of a subset $\mathcal S$ of $\mathcal P$ can be found in $\mathcal O(4^{|{\mathcal S}|} \cdot |\mathcal S| ^ k)$ time.
\end{lemma}

\begin{proof}
We observe that for any $\mathcal S \subseteq \mathcal P$, $\mathcal S$ forms a convex polygon with at most $k$ points lying in the interior. For $|{\mathcal S}| \le 2$, the statement of the lemma is trivially true. Hence we assume that $|{\mathcal S}| > 2$.

From \cite{hwang1992steiner}, the number of full Steiner topologies of $\mathcal{S}$ is 

$$\frac{|\mathcal S|!}{|\mathrm{CH}(\mathcal{S})|!} \cdot \frac {\binom{2|{\mathcal S}| - 4}{|{\mathcal S}| - 2}}{|{\mathcal S}| - 1}$$

However, we know that:

$$\frac {\binom{2|{\mathcal S}| - 4}{|{\mathcal S}| - 2}}{|{\mathcal S}| - 1} < \frac{\binom{2|{\mathcal S}|}{|{\mathcal S}|}}{|{\mathcal S}|} < \frac{\sum \limits_{r = 0}^{2|{\mathcal S}|} \binom{2|{\mathcal S}|}{k}}{|{\mathcal S}|} = \frac{2 ^{2|{\mathcal S}|}}{|{\mathcal S}|} = \frac{4^{|{\mathcal S}|}}{|{\mathcal S}|}$$

And,

$$\frac{|\mathcal S|!}{|\mathrm{CH}(\mathcal{S})|!} < \frac{|\mathcal S|!}{(|\mathcal S| - k)!} < |\mathcal S|^k$$

Therefore, the number of full Steiner topologies of $\mathcal{S}$ is at most $4^{|{\mathcal S}|} |\mathcal S| ^ {k - 1}$. Each of these topologies can be enumerated and using \emph{Melzak's FST Algorithm}, we can also find the SMT realizing each such full Steiner topology in linear time, as given in \cite{hwang1992steiner}. Therefore to iterate over all topologies and find a minimum takes at most time:

$$(4^{|{\mathcal S}|} \cdot |\mathcal S| ^ {k - 1}) \cdot \mathcal O(|\mathcal S|) = \mathcal \mathcal O(4^{|{\mathcal S}|} \cdot|\mathcal S| ^ k)$$

\end{proof}

Now, we find the time required for extending the results of Lemma~\ref{four_power_n} to all subsets of $\mathcal{P}$.
\begin{lemma}
\label{five_power_n}
Let $\mathcal P$ be a $k$-Almost Convex Set. Computing a minimum FST \textbf{for all} subsets $\mathcal S \subseteq \mathcal P$ can be done in $\mathcal O(n^k \cdot 5^{n})$ time.
\end{lemma}
\begin{proof}

Using Lemma \ref{four_power_n}, we can get a minimum FST for a single subset $\mathcal S \subseteq \mathcal P$ in $\mathcal O(4^{|{\mathcal S}|} |\mathcal S| ^ k)$ time. Moreover, we know that the number of subsets of $\mathcal P$ that are of size $r$ is $\binom n r$. This means that the total time to compute a minimum FST \textbf{for all} subsets $\mathcal S \subseteq \mathcal P$, time taken is:

$$\sum \limits_{r = 0} ^ {n} \binom{n}{r} \cdot \mathcal O(4^r \cdot r^k) = \sum \limits_{r = 0} ^ {n} \binom{n}{r} \cdot \mathcal O(n^k \cdot 4^r) = \mathcal O(n^k \cdot (1 + 4)^n) = \mathcal O(n^k \cdot 5^n)$$

\end{proof}

For each $\mathcal S \subseteq \mathcal P$, we denote by $\mathcal F_{\mathcal S} $ a minimum FST of $\mathcal S$ and by $\mathcal T_{\mathcal S} $ the SMT of $\mathcal S$.

\begin{lemma}
\label{single_subset_SMT}
The SMT of subset $\mathcal S \subseteq \mathcal P$, $\mathcal T_{\mathcal S}$, can be found in $\mathcal O(|{\mathcal S}| \cdot 2^{|{\mathcal S}|})$ time, given that we have pre-computed $\mathcal T_{\mathcal R}$ and $\mathcal F_{\mathcal R}$, $\forall \mathcal R \subseteq \mathcal S$.
\end{lemma}

\begin{proof}

If $\mathcal T_{\mathcal S}$ was a full Steiner tree then it would be $\mathcal F_{\mathcal S}$. Otherwise, $\mathcal T_{\mathcal S}$ contains multiple full Steiner subtrees.

Let $\mathcal F$ be a leaf full Steiner subtree of $\mathcal T_{\mathcal{S}}$ with pivot $P_{\mathcal F}$. Therefore from Lemma \ref{leaf_non_leaf_structure} we have $\mathcal T_{\mathcal S} = \mathcal T_{((\mathcal S - V(\mathcal F)) \cup \{P_{\mathcal F}\})} \cup \mathcal F$. Therefore we can iterate over all subsets $\mathcal R \subset \mathcal S$ and all terminals $P \in \mathcal{R}$, and take the minimum-length tree among $\mathcal T_{((\mathcal S - \mathcal R)  \cup \{P\})} \cup \mathcal F_{\mathcal R}$. Since we are iterating over all $\mathcal R \subset \mathcal S$, and all $P \in \mathcal R$, we are guaranteed to get $\mathcal R = V(\mathcal F) \cap \mathcal{S}$ and $P = P_{\mathcal F}$ on one such iteration.

Now, as there are $\mathcal O(2^{|\mathcal S|})$ possibilities of $\mathcal R \subset \mathcal S$ and $\mathcal O(|\mathcal S|)$ possibilities of $P \in \mathcal R$, we have $\mathcal O(|S| \cdot 2^{|\mathcal S|})$ possibilities of the pair $(\mathcal R, P)$. Therefore the total time required for iterating is $\mathcal O(|{\mathcal S}| \cdot 2^{|{\mathcal S}|})$.

\end{proof}

\begin{lemma}
\label{multi_subset_SMT}
SMTs for all subsets $\mathcal S \subseteq \mathcal P$, $\mathcal T_{\mathcal S}$ can be found in $\mathcal O(n \cdot 3 ^n)$ time, given that we have precomputed a minimum FST $\mathcal F_{\mathcal S}$ $\forall \mathcal S \subseteq \mathcal P$.
\end{lemma}

\begin{proof}

Using Lemma \ref{single_subset_SMT}, we can get the SMT $\mathcal T_{\mathcal S}$, for a single subset $\mathcal S \subseteq \mathcal P$ in $\mathcal O(r \cdot 2^r)$ time, where $|{\mathcal S}| = r$. Moreover, we know that the number of subsets of $\mathcal P$ that are of size $r$ is $\binom n r$. This means that the total time to compute the SMT for all subsets $\mathcal S \subseteq \mathcal P$, time taken is: 

$$\sum \limits_{k = 0} ^ {n} \binom{n}{k} \cdot \mathcal O(k \cdot 2^k) = \mathcal O(n \cdot (1 + 2)^n) = \mathcal O(n \cdot 3^n)$$

However, to apply Lemma \ref{single_subset_SMT} on some subset $\mathcal S \subseteq \mathcal P$ for computing $\mathcal T_{\mathcal S}$, we must also have $\mathcal T_{\mathcal R}$ precomputed for all $\mathcal R \subseteq \mathcal S$. This can be guaranteed by computing $\mathcal T_{\mathcal S}$ and $\mathcal{F}_{\mathcal{S}}$, for all subsets $\mathcal S \subseteq \mathcal P$ in an increasing order of $|\mathcal S|$ (or any order which guarantees that the subsets of $\mathcal S$ are processed before $\mathcal S$).

\end{proof}

Finally, we state our algorithm.

\begin{theorem}
\label{final_theorem}

An SMT $\mathcal T_{\mathcal P}$ of a $k$-Almost Convex Set $\mathcal P$ of terminals can be computed in $\mathcal O(n^k \cdot 5^n)$ time.

\end{theorem}

\begin{proof}

Consider the following algorithm: 
\begin{algorithm}[H]
\caption{Computation of $\mathcal{T_P}$ ~~~ \textbf{Input:} $\mathcal P$ }\label{alg:main algo}
\begin{algorithmic}[1]
\For{all $\mathcal S \subseteq \mathcal P$} 
\State Compute $\mathcal {F_S}$ \Comment{Using Lemma \ref{five_power_n}}
\EndFor \Comment{This takes $\mathcal O(n^k \cdot 5^n)$ time}
\For{all $\mathcal S \subseteq \mathcal P$}
\State Compute $\mathcal {T_S}$ \Comment{Using Lemma \ref{multi_subset_SMT}}
\EndFor \Comment{This takes $\mathcal O(n \cdot 3^n)$ time}
\State \Return $\mathcal{T_P}$ \Comment{Total runtime is $\mathcal O(n^k \cdot 5^n + n \cdot 3^n) = \mathcal O(n^k \cdot 5^n)$}
\end{algorithmic}
\end{algorithm}


Hence we have an SMT of a $k$-Almost Convex Point Set $\mathcal P$ in $\mathcal O(n^k\cdot 5^n)$ time.

\end{proof}

The above theorem gives us several improvements in special classes of inputs, based on the number of input points lying inside the convex hull of the input set, as described in the following corollary. Let there be an $f(n)$-Almost Convex Point Set $\mathcal P$ containing $n$ points. Recall that $\mathcal{P} = \mathcal{P}_1 \uplus \mathcal{P}_2$, $\mathcal{P}_1$ containing the points on the convex hull of $\mathcal{P}$, and $|\mathcal{P}_2| = f(n)$. It is only possible that $f(n) \leq n$. 

\begin{corollary}\label{almost-better}
Let $\mathcal P$ be a $f(n)$-Almost Convex Point Set. Then, then there is an algorithm $\mathcal{A}$ for \ESMT such that, $\mathcal{A}$ runs in $2^{\OO{(n + f(n) \log n)}}$ time. In particular, 
\begin{enumerate}
    \item When $f(n) = \OO(n)$, $\mathcal A$ runs in $2^{\OO(n\log n)}$ time.
    \item When $f(n) = \Omega(\frac{n}{\log n})$ and $f(n) = o(n)$, $\mathcal{A}$ runs in $2^{o(n\log n)}$.
    \item When $f(n) = \OO(\frac{n}{\log n})$, $\mathcal A$ runs in $2^{\OO(n)}$ time.
  \end{enumerate}
\end{corollary}

Therefore, for $f(n) = o(n)$, our algorithm for \ESMT does better on $f(n)$-Almost Convex Points Sets than the current best known algorithm~\cite{hwang1986linear}.
    
    

\section{Approximation Algorithms for \ESMT}\label{sec:apx_esmt}

The \ESMT problem is NP-hard as shown by Garey et al. in~\cite{garey1977complexity}. Garey et al. also prove that there cannot be an FPTAS (fully polynomial time approximation scheme) for this problem unless $P=NP$. At the same time, the case when all the terminals lie on the boundary of a convex region admits an FPTAS as given in~\cite{scott1988convexity}. We aim to conduct a more fine-grained analysis for the problem by considering $f(n)$-Almost Convex Point Sets of $n$ terminals and studying the existence of FPTASes for different functions $f(n)$. First, we present an FPTAS for \ESMT on $f(n)$-Almost Convex Sets of $n$ terminals, when $f(n) = \OO (\log n)$. Next, we prove that no FPTAS exists for the case when $f(n) =\Omega (n^\epsilon)$, where $\epsilon \in (0,1]$.

\subsection{FPTAS for \ESMT on Cases of Almost Convex Point Sets}\label{subsec:fptas}

We first propose an algorithm for computing the SMT of a planar graph $G$ having $N$ vertices and $n$ terminals, out of which $k$ terminals lie on the outer face of $G$ and the remaining terminals lie within the boundary. Next, following the procedure in~\cite{scott1988convexity} we get an FPTAS for \ESMT on $f(n)$-Almost Convex Sets of $n$ terminals, where $f(n) = \OO (\log n)$.

We state the following proposition from Theorem 1 in~\cite{scott1988convexity}:
\begin{proposition}\label{prop:tree_interval}
    Let $\C{P}$ be the vertices of any polygon in the plane, $\C{K}$ a subset of $\C{P}$, and $\C{T}$ a tree consisting of all the vertices of $\C{K}$ (and possibly some other vertices as well) and contained entirely inside $\C{P}$. Then on removing any edge of the tree, we get two disjoint trees $\mathcal{T}_1, \mathcal{T}_2$, such that the vertices of $\C{K}$ in each tree $\mathcal{T}_i, i \in \{1,2\}$ form an interval in $\C{K}$.
\end{proposition}

Using~\Cref{prop:tree_interval} and the Dreyfus-Wagner algorithm~\cite{dreyfus1971steiner}, we give an algorithm for obtaining the SMT of a planar graph $G$. 

Let $\C{K}$ represent the set of terminals lying on the outer face of $G$ and $\C{R}$ be the set of terminals lying inside the outer face of $G$. We have $|V(G)|=N$, $|\C{K} \cup \C{R}|=n$, and $|\C{K}|=k$. Let $C(\C{L})$ denote the SMT in $G$ for a terminal subset $\C{L} \subseteq V(G)$. Let $B(v,\C{L},[a,b))$ denote the SMT in $G$ for the terminal set $\{v\} \cup \C{L} \cup [a,b)$, where $\C{L} \subseteq \C{R}$, $[a,b)$ is the set of vertices in $\C{K}$ forming an interval from vertex $a$ to $b$ in counterclockwise direction along the outer boundary of $G$ including $a$ but excluding $b$, $v \in V(G) \setminus (\C{L} \cup [a,b))$, and the degree of $v$ is at least $2$ in $B(v,\C{L},[a,b))$. Let $A(v,\C{L},[a,b))$ denote the SMT in $G$ for the terminal set $\{v\} \cup \C{L} \cup [a,b)$, where $\C{L}$, $[a,b)$, and $v$ are as defined in the previous case, and the degree of $v$ is at most $1$ in $A(v,\C{L},[a,b))$.

Splitting the SMT at a vertex $v$ of degree at least $2$ gives rise to two smaller instances of the {\sc Steiner Minimal Tree} problem on graphs.

\begin{equation}\label{eq:dw_B}
    B(v,\C{L},[a,b)) = \min_{\Pi_1, \Pi_2, \Pi_3} \{C(\{v\} \cup \C{L}' \cup [a,x)) + C(\{v\} \cup (\C{L}\setminus \C{L}') \cup [x,b)) \}
\end{equation} where the conditions on $\C{L}'$ and $x$ are $\Pi_1: \C{L}' \subseteq \C{L}$, $\Pi_2: x \in \mathcal{K}, a< x <b$, and $\Pi_3: \emptyset \subset \C{L}' \cup [a,x) \subset \C{L} \cup [a,b)$.

The intuition is to root the tree at an internal terminal vertex and start growing the Steiner tree from there. Observe that on removing one of the internal vertices $v$ in the tree $\C{T}$, we get one, two or three disjoint subtrees. They induce a partition over the terminals. The terminals in $\C{K}$ in each of the subtrees form intervals in $\mathcal{K}$, according to~\Cref{prop:tree_interval}. Moreover, the terminals in $\C{R}$ can be partitioned in any way, not necessarily maintaining the interval structure. This is captured in the following recurrence relation:
\begin{equation}\label{eq:dw_C}
    C(\{v\} \cup \C{L} \cup [a,b)) = \min_{\Pi_1, \Pi_2, \Pi_3}\{A(v,\C{L}_1,[a,c))+A(v,\C{L}_2,[c,d))+A(v,\C{L}_3,[d,b))\}
\end{equation}

where the conditions on $\C{L}_1$, $\C{L}_2$, $\C{L}_3$, $c$, and $d$ are $\Pi_1: \C{L}_1,\C{L}_2,\C{L}_3 \subseteq \C{L}$, $\Pi_2: \C{L}_1 \cup \C{L}_2 \cup \C{L}_3 = \C{L}$, and $\Pi_3: c,d \in \mathcal{K}, a\leq c \leq d \leq b$ and we have 
\begin{equation}\label{eq:dw_A}
    A(v,\C{L}',[p,q))=\min\{\min_{u \notin \C{L}'}\{B(u,\C{L}',[p,q))+d(u,v)\},\min_{u \in \C{L}' \cup [p,q)}\{C(\C{L}' \cup [p,q))+d(u,v)\}\}
\end{equation}

Our aim is to compute $C(\{v\} \cup (\C{R}\setminus \{v\}) \cup \C{K})$, where $v \in \C{R}$. We precompute the shortest distance between all pairs of vertices. We then compute the values of $C(.)$ and $B(.)$ in increasing order of cardinality of subsets of vertices in $\C{K}$ and $\C{R}$. Let $d(u,v)$ denote the shortest path length between $u$ and $v$. The base cases are $C(\{v\} \cup \{a\}) = d(v,a)$ for all $v \in V(G)$ and $a \in \C{K} \cup \C{R}$.

\begin{algorithm}[H]
\caption{Computation of SMT of planar graph $G$ with terminal set $\C{K} \cup \C{R}$ ~~~ \textbf{Input:} $G$, $\C{K}$, $\C{R}$}\label{alg:smt_planar}
\begin{algorithmic}[1]
\State Compute the shortest distance between all pairs of vertices
\For{all $u \in V(G)$ and $a \in K \cup R$} 
\State Set $C(\{u\} \cup \{a\}) = d(u,a)$
\EndFor
\State Select a vertex $v \in \C{R}$
\For{$i = 1, \ldots, n-k-1$}
\For{each $\C{L} \subseteq \C{R}\setminus \{v\}$ of size $i$}
\For{$j = 1, \ldots, k$}
\If{j=k}
\State Compute $B(v,\C{L},\C{K})$ using~\Cref{eq:dw_B}
\State Compute $C(\{v\} \cup \C{L} \cup \C{K})$ using~\Cref{eq:dw_C}
\Else
\For{each $[a,b) \subseteq \C{K}$ of size $j$}
\State Compute $B(v,\C{L},[a,b))$ using~\Cref{eq:dw_B}
\State Compute $C(\{v\} \cup \C{L} \cup [a,b))$ using~\Cref{eq:dw_C}
\EndFor
\EndIf
\EndFor
\EndFor
\EndFor
\State \Return $C(\{v\} \cup (\C{R}\setminus \{v\}) \cup \C{K})$
\end{algorithmic}
\end{algorithm}

We analyse the correctness and running time of~\Cref{alg:smt_planar}.
    \begin{theorem} \label{thm:algo1_correct}
     Consider a planar graph $G$ on $N$ vertices and a set $\mathcal{K} \uplus \mathcal{R} \subseteq V(G)$ of $n$ terminals such that $\mathcal{K}$ is defined as the terminals lying on the outer face of $G$. Moreover, let $|\mathcal{K}| = k$. Then~\Cref{alg:smt_planar} computes the SMT for $\mathcal{K} \uplus \mathcal{R}$ in $G$ in time $\OO(N^2k^44^{n-k} + Nk^33^{n-k}+N^3)$.
    \end{theorem}
    \begin{proof}
{\bf Correctness of~\Cref{alg:smt_planar}.} 
In order to prove the correctness of~\Cref{alg:smt_planar}, we need to show that the~\Cref{eq:dw_B,eq:dw_C} are valid.

        In~\Cref{eq:dw_B}, $B(v,\C{L},[a,b))$ denotes an SMT for the terminal set $\{v\} \cup \C{L} \cup [a,b)$, conditioned on the fact that the degree of $v$ is at least $2$ in it. Let us split the SMT at vertex $v$ into two smaller subtrees. This must also split the terminals in $[a,b)$ in two intervals $[a,x)$ and $[x,b)$, respectively. Otherwise it would mean that the SMT has crossing edges, which is not possible. The vertices in $\mathcal{L}$ can be present in any of the two subtrees, hence we consider all possible partitions of $\C{L}$ into two subsets $\C{L}'$ and $\C{L} \setminus \C{L}'$. Thus for~\Cref{eq:dw_B}, $\mbox{LHS} \geq \mbox{RHS}$. On the other hand, the expression in the RHS of~\Cref{eq:dw_B} is a tree containing the vertex subset $\{v\} \cup \mathcal{L} \cup [a,b)$. Since, $B(v,\mathcal{L},[a,b))$ is an SMT for the same vertex subset, in~\Cref{eq:dw_B}  $\mbox{LHS} \geq \mbox{RHS}$. Therefore,~\Cref{eq:dw_B} is valid.

        For~\Cref{eq:dw_C}, we take $v$ as the root of the SMT $C(\{v\} \cup \C{L} \cup [a,b))$. The degree of $v$ in the SMT can be $1$, $2$, or $3$. Accordingly, on removing $v$, we will get $1$, $2$, or $3$ subtrees, with the condition that in each subtree the degree of $v$ is $1$. Again, the terminals in $[a,b)$ are divided into smaller intervals $[a,c)$, $[c,d)$ and $[d,b)$ for some $c,d \in \mathcal{K}$ satisfying $a<c<d<b$. The terminals in $\C{L}$ are divided among the subtrees in any combination. The number of such intervals and partitions is equal to the degree of $v$ in $C(\{v\} \cup \C{L} \cup [a,b))$. The term $A(v,\C{L}',[p,q))$ is obtained by minimizing across all SMTs satisfying the condition that degree of $v$ in the SMT is $1$. Thus in~\Cref{eq:dw_C}, $\mbox{LHS} \geq \mbox{RHS}$. On the other hand, the RHS of~\Cref{eq:dw_C} given a tree containing vertices $\{v\} \cup \mathcal{L} \cup [a,b)$. Since $C(\{v\} \cup \mathcal{L} \cup [a,b))$ is an SMT on the terminal set $\{v\} \cup \mathcal{L} \cup [a,b)$, in~\Cref{eq:dw_C} $\mbox{LHS} \leq \mbox{RHS}$. Thus,~\Cref{eq:dw_C} is valid.

{\bf Running time of~\Cref{alg:smt_planar}.}
All pairs shortest paths can be calculated in $\OO(N^3)$ time. The time complexity of the dynamic program has two components to it. One is due to computation of the B(.) values using~\Cref{eq:dw_B} and the other is for calculating the C(.) values using~\Cref{eq:dw_C}.
    \begin{enumerate}
        \item The number of computational steps for calculating $B(v,\mathcal{L},[a,b))$ using~\Cref{eq:dw_B} is of the order of the number of choices of $v$, $\mathcal{L}$, $\mathcal{L}'$, $[a,b)$, and $[a,x)$ such that $\mathcal{L} \subset \mathcal{R}$, $\mathcal{L}' \subseteq \mathcal{L}$, $a,x,b \in \mathcal{K}$, $a \leq x \leq b$, and $v \in V(G) \setminus (L \cup [a,b))$. Each vertex in $\mathcal{R}$ belongs to exactly one of the sets $\mathcal{L}'$, $\mathcal{L}\setminus \mathcal{L}'$, or $V(G) \setminus \mathcal{L}$. The vertices in $\mathcal{K}$ are partitioned into three intervals, $[a,x)$, $[x,b)$, and $[b,a)$. There are at most $N$ possibilities for $v$. This gives us a running time of $\OO(Nk^3 3^{n-k})$.
        \item The number of computational steps for calculating $C(\{v\} \cup \mathcal{L} \cup [a,b))$ using~\Cref{eq:dw_C} is $3N$ times the order of the number of choices of $v$, $\mathcal{L}$, $\mathcal{L}_1$, $\mathcal{L}_2$, $a$, $b$, $c$, and $d$ such that $\mathcal{L} \subset \mathcal{R}$, $\mathcal{L}_1 \subseteq \mathcal{L}$, $\mathcal{L}_2 \subseteq \mathcal{L}$, $\{a,c,d,b\} \subseteq \mathcal{K}$, $a \leq c \leq d \leq b$, and $v \in V(G) \setminus (\mathcal{L} \cup [a,b))$. Each vertex in $\mathcal{R}$ belongs to exactly one of the sets $\mathcal{L}_1$, $\mathcal{L}_2$, $\mathcal{L}_3$ or $V(G) \setminus \mathcal{L}$. The vertices in $\mathcal{K}$ are partitioned into at most four intervals, $[a,c)$, $[c,d)$, $[d,b)$ and $[b,a)$. There are at most $N$ possibilities for $v$. The $3N$ factor is because calculating each of the $C(.)$ values involves minimization over at most $3N$ terms. This gives us a running time of $\OO(N^2k^4 4^{n-k})$.
    \end{enumerate}
    Thus, the time complexity of the algorithm is $\OO(N^2k^44^{n-k}+Nk^33^{n-k}+N^3)$.
\end{proof}
We obtain the following corollary from the above theorem.

\begin{corollary}\label{graph-correct}
 Consider a planar graph $G$ on $N$ vertices and a set $\mathcal{K} \uplus \mathcal{R} \subseteq V(G)$ of $n$ terminals such that $\mathcal{K}$ is defined by the terminals lying on the outer face of $G$. Moreover, let $|\mathcal{K}| = k$ and let $|\C{R}| = n-k = \OO(\log n)$. Then~\Cref{alg:smt_planar} computes the SMT for $\mathcal{K} \uplus \mathcal{R}$ in $G$ in time $N^3k^4n^{\OO(1)}$.
\end{corollary}

Next we state the FPTAS for \ESMT on $f(n)$-Almost Convex Sets of $n$  terminals. This is achieved by converting the instance of \ESMT into an instance of {\sc Steiner Minimal Tree} on graphs. The {\sc Steiner Minimal Tree} problem shall be solved using~\Cref{alg:smt_planar}. For this, we use the Algorithm 2 in~\cite{scott1988convexity}. We restate the Algorithm 2 in~\cite{scott1988convexity} for our problem instance. We denote the set of terminals with $\C{P}$.

\begin{algorithm}[H]
\caption{Computation of $(1+\epsilon)$-approximate SMT of $\C{P}$ ~~~ \textbf{Input:} $\mathcal P, \epsilon$ }\label{alg:smt_fptas}
\begin{algorithmic}[1]
\State Compute the convex hull of the set of terminals $\C{P}$. Let the region enclosed by the convex hull $\mathrm{CH}(\C{P})$ be denoted by $\mathbb{R}_{\mathrm{CH}(\C{P})}$. Let the points in $\C{P}$ lying on $\mathrm{CH}(\C{P})$ be $\C{K}$ and $\C{R} = \C{P} \setminus \C{K}$.
\State We enclose the set of terminals $\C{P}$ with the smallest axis-parallel bounding square. Let its side length be $D$. We divide the bounding square into same sized grids of side length $\frac{D\epsilon}{8n-12}$, where $\epsilon$ is the approximation factor.
\State Let $\C{V}_0$ be the set of all lattice points introduced in the previous step, and $\C{V}_1$ be the set of all lattice points lying on the edges of $\mathrm{CH}(\C{P})$. We define the weighted graph $G_{f,\epsilon}$ to be the complete graph with vertex set $V(G_{f,\epsilon}) = \C{K} \cup \C{R} \cup (\C{V}_0 \cap \mathbb{R}_{\mathrm{CH}(\C{P})}) \cup \C{V}_1$. The edge weights are equal to the Euclidean distance between the two end points.
\State Return the SMT $\C{T}$ for the graph $G_{f,\epsilon}$ with $\C{K} \cup \C{R}$ as the terminal set using~\Cref{alg:smt_planar}. 
\end{algorithmic}
\end{algorithm}

We analyse the correctness and running time of~\Cref{alg:smt_fptas}.
    \begin{theorem} \label{thm:algo2_correct}
     Consider a set $\C{P}$ of $n$ points such that $\mathcal{K}$ is defined as the points lying on the convex hull of $\C{P}$, i.e. $\mathrm{CH}(\C{P})$, and $\C{R} = \C{P} \setminus \C{K}$. Moreover, let $|\mathcal{K}| = k$. Then~\Cref{alg:smt_fptas} computes a $(1+\epsilon)$-approximate SMT for $\mathcal{P}$ in time $\OO(\frac{n^4k^4}{\epsilon ^4}4^{n-k})$.
    \end{theorem}
    \begin{proof}
{\bf Correctness of~\Cref{alg:smt_fptas}.} 
    In order to prove the correctness of~\Cref{alg:smt_fptas}, we need to show that $\C{T}$ is a $(1+\epsilon)$-approximation of the SMT of the terminal set $\C{P}$, and $\C{T}$ is indeed the SMT for $\C{K} \cup \C{R}$ in the complete weighted graph $G_{f,\epsilon}$.

        In~\cite{scott1988convexity}, the concept of weight planar graphs is used. A graph $G$ is called weight planar if it is a non-planar graph embedded on the Euclidean plane, having non-negative edge weights, such that every pair of edges $(u,v)$ and $(u',v')$ which intersect in this embedding of $G$, satisfy the inequality: $w(u,v) + w(u',v') > d(u,u') + d(v,v')$, where $w(u,v)$ is the weight of the edge between vertices $u$ and $v$ and $d(x,y)$ is the length of the shortest path between vertices $x$ and $y$. Since the edge weights of $G_{f,\epsilon}$ are the Euclidean distances between the points, $G_{f,\epsilon}$ is a weight planar graph. 

        From Theorem 5 of \cite{scott1988convexity}, we get that the SMT of a weight planar graph does not contain any crossing edges even though the input graph is non-planar. Because the SMT does not contain any crossing edges and lies completely inside the convex hull of the terminal pointset, the terminals on the outer boundary of $G_{f,\epsilon}$, i.e. $\C{K}$, follow the interval pattern as stated in~\Cref{prop:tree_interval}. Therefore,~\Cref{alg:smt_planar} designed for planar graphs can be applied in the case of weight planar graphs as well. So, $\C{T}$ is the SMT for $\C{K} \cup \C{R}$ in the complete weighted graph $G_{f,\epsilon}$.

        Finally, from Theorem 12 in~\cite{scott1988convexity}, we get that the length of the Steiner tree obtained from~\Cref{alg:smt_fptas} is at most $(1+\epsilon)$ times the length of the SMT $\C{T}^*$ of $\C{P}$, i.e. $|\C{T}| \leq (1+\epsilon)|\C{T}^*|$. Thus, we are done.
        
{\bf Running time of~\Cref{alg:smt_fptas}.}
    Constructing the convex hull takes $\OO(n\log n)$ time. The number of lattice points contained in the bounding box is $\big(\frac{8n-12}{\epsilon}\big)^2 = \OO(\frac{n^2}{\epsilon ^2})$. Thus, the number of vertices in the resultant graph $G_{f,\epsilon}$ is $N = \OO(\frac{n^2}{\epsilon ^2}) + n = \OO(\frac{n^2}{\epsilon ^2})$. The time complexity of~\Cref{alg:smt_planar} is $\OO(N^2k^44^{n-k}+Nk^33^{n-k}+N^3) = \OO(\frac{n^4k^4}{\epsilon ^4}4^{n-k})$. This step dominates the running time resulting in the complexity of~\Cref{alg:smt_fptas} being $\OO(\frac{n^4k^4}{\epsilon ^4}4^{n-k})$.
    \end{proof}

\begin{theorem}\label{thm:esmt_fptas}
    There exists an FPTAS for \ESMT on an $f(n)$-Almost Convex Set of $n$ terminals, where $f(n) = \OO (\log n)$.
\end{theorem}

\begin{proof}
    From~\Cref{thm:algo2_correct}, we get a $(1+\epsilon)$-approximate SMT for any $(n-k)$-Almost Convex Set of $n$ terminals in time $\OO(\frac{n^4k^4}{\epsilon ^4}4^{n-k})$. For $n-k = \OO(\log n)$, we get the running time of~\Cref{alg:smt_fptas} to be $\OO(\frac{n^4k^4}{\epsilon ^4}n^{\OO(1)})$. Thus,~\Cref{alg:smt_fptas} is an FPTAS for the Euclidean Steiner Minimal Tree problem on an $\OO (\log n)$-Almost Convex Set of $n$ terminals. 
\end{proof}

\subsection{Hardness of Approximation for \ESMT on Cases of Almost Convex Sets}\label{subsec:apx_hardness}

In this section, we consider the \ESMT problem on $f(n)$-Almost Convex Sets of $n$ terminal points, where $f(n) =\Omega(n^\epsilon)$ for some $\epsilon \in (0,1]$. We show that this problem cannot have an FPTAS. The proof strategy is similar to that in~\cite{garey1977complexity}. First, we give a reduction for the problem \textsc{Exact Cover by $3$-Sets} (defined below) to our problem to show that our problem is NP-hard. Next, we consider a discrete version of our problem and reduce our problem to the discrete version. The discrete version is in NP. Therefore, this chain of reductions imply that the discrete version of our problem is Strongly NP-complete and therefore cannot have an FPTAS, following from~\cite{garey1977complexity}. Similar to the arguments in~\cite{garey1977complexity}, this also implies that our problem cannot have an FPTAS.

Before we describe our reductions, we take a look at the NP-hardness reduction of the \ESMT problem from the \textsc{Exact Cover by 3-Sets} (X3C) problem in~\cite{garey1977complexity}. In the X3C problem, we are given a universe of elements $U = \{1, 2, \ldots, 3n\}$ and a family $\mathbb{F}$ of $3$-element subsets $F_1, F_2, \ldots, F_t$ of the $3n$ elements. The objective is to decide if there exists a subcollection $\mathbb{F}' \subseteq \mathbb{F}$ such that: (i) the elements of $\mathbb{F}'$ are disjoint, and (ii) $\bigcup_{F' \in \mathbb{F}'} F' = U$. The X3C problem is NP-complete~\cite{garey1979computers}.

In~\cite{garey1977complexity}, various gadgets are constructed, i.e.~particular arrangements of a set of points. These are then arranged on the plane in a way corresponding to the given X3C instance.~\Cref{fig:redcn} shows the reduced ESMT instance obtained for $U = \{1,2,3,4,5,6\}$ and $\mathbb{F} = \{\{1, 2, 4\}, \{2, 3, 6\}, \{3, 5, 6\}\}$ (taken from~\cite{garey1977complexity}). The squares, hexagons (crossovers), shaded circles (terminators) and lines (rows) all represent specific arrangements of a subset of points. Let $X(\mathbb{F})$ denote the reduced instance. The number of points in $X(\mathbb{F})$ is bounded by a polynomial in $n$ and $t$. Let this polynomial be $\OO(t^\gamma)$, as we can assume $t \geq n$ since otherwise it trivially becomes a NO instance. Here $\gamma$ is some constant.

\begin{figure}[h]
\centering
\subfloat{\includegraphics[width=12cm]{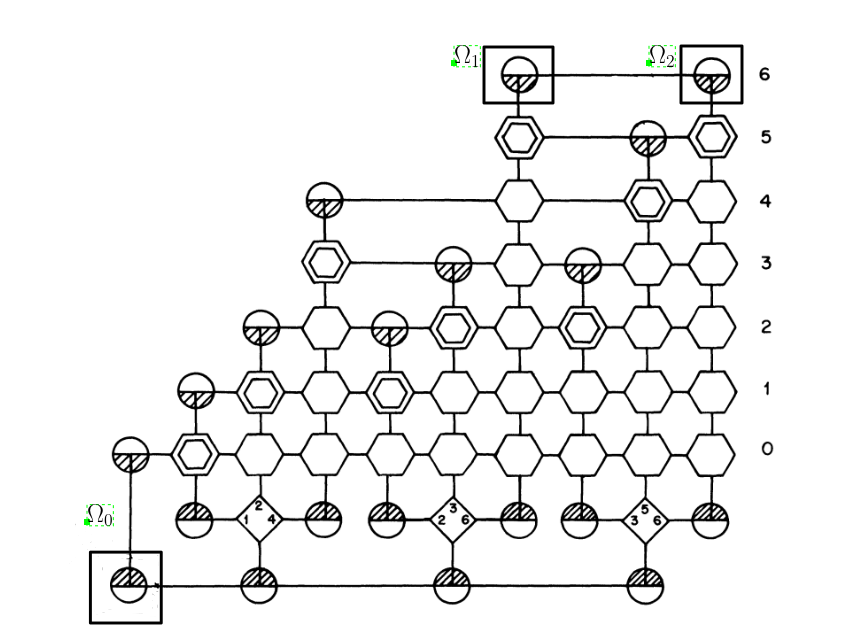}}
\caption{Reduced instance of ESMT from X3C (taken from~\cite{garey1977complexity})}
\label{fig:redcn}
\end{figure}

We restate Theorem 1 in~\cite{garey1977complexity}.
\begin{proposition}\label{thm:redcn}
    Let $\mathcal{S}^{*}$ denote an SMT of $X(\mathbb{F})$, the instance obtained by reducing the X3C instance $(n,\mathbb{F})$, and $|\mathcal{S}^{*}|$ denote its length. If $\mathbb{F}$ has an exact cover, then $|\mathcal{S}^{*}| \leq f(n,t,\hat{C})$, otherwise $|\mathcal{S}^{*}| \geq f(n,t,\hat{C}) + \frac{1}{200nt}$, where $t = |\mathbb{F}|$, $\hat{C}$ is the number of crossovers, i.e.~hexagonal gadgets, and $f$ is a positive real-valued function of $n,t,\hat{C}$.
\end{proposition}

We extend this construction to prove NP-hardness for instances of \ESMT where the terminal set $\mathcal{P}$ has $\Omega(n^\epsilon)$ points inside $\mathrm{CH}(\mathcal{P})$. Here, $\epsilon \in (0,1]$ and $n$ is the number of terminals.

Let us call the \emph{length} of a gadget to be the maximum horizontal distance between any two points in that gadget. Similarly, we define the \emph{breadth} of a gadget to be the maximum vertical distance between any two points in that gadget.




\begin{figure}[h] 
\centering
\subfloat[\centering The Upward Terminator symbol]{\includegraphics[width=1.65cm]{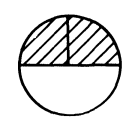}}   \qquad\qquad
\subfloat[\centering The Downward Terminator symbol]{\includegraphics[width=1.5cm]{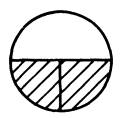}}
    \qquad\qquad
\subfloat[\centering The Terminator gadget] {\includegraphics[width=5cm]{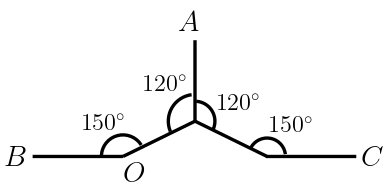} }
\caption{The Terminator gadget symbol and arrangement of points}  
\label{fig:terminator}
\end{figure}

The \emph{terminator} gadget used is shown in~\Cref{fig:terminator}. The straight lines represent a row of at least $1000$ points separated at distances of $1/10$ or $1/11$. The angles between them are as shown. The upward terminator has the point $A$ above the other points in the terminator, whereas the downward terminator has the point $A$ below the other points. Firstly, we adjust the number of points in the long rows, such that the length and breadth of the terminators is same as that of the hexagonal gadgets (crossovers). We can fix this length and breadth to be some constants, such that the number of points in each gadget is also bounded by some constant. In our construction, we modify the terminators $\Omega_0$, $\Omega_1$, and $\Omega_2$ as shown in~\Cref{fig:redcn} enclosed in squares. $\Omega_1$ is the terminator corresponding to the first occurrence of the element $3n\in U$ in some set in $\mathbb{F}$ and $\Omega_2$ is the terminator corresponding to the last occurrence of $3n$ in some set in $\mathbb{F}$ (if there are more than one occurrences of $3n$). If there are no occurrences of $3n$, then it is trivially a no-instance. The modified gadgets are shown in~\Cref{fig:terminator_new}. All the other gadgets remain unaltered.

\begin{figure}[h]
\centering
\includegraphics[width=8cm]{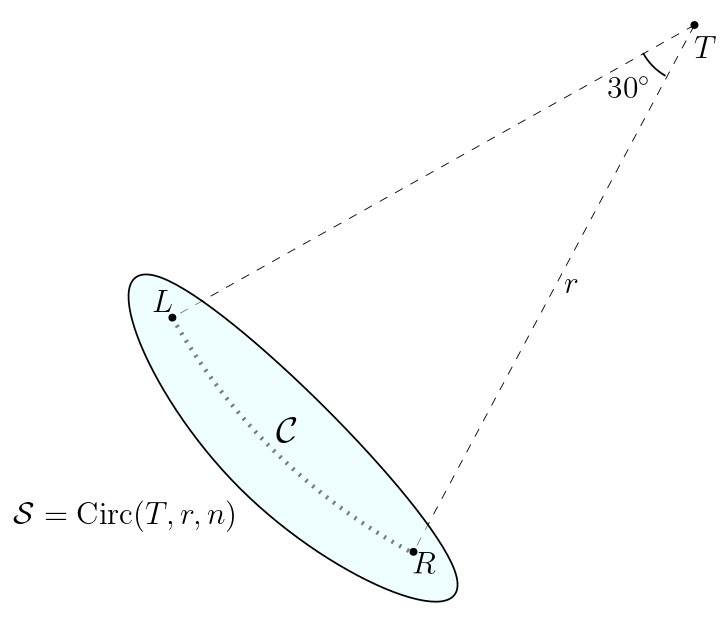}
\caption{Conic Set: $\mathrm{Cone}(T,r,n)$}
\label{fig:conic_set}
\end{figure}

We call a set of points arranged as shown in~\Cref{fig:conic_set}, as a Conic Set.
\begin{definition} \label{def:conic_set}
    A Conic Set is a set of points consisting of a point $T$, called the tip of the cone, and the remaining points denoted by $\mathcal{S}$. Let $\mathcal{C}$ be the circle with $T$ as centre and radius $r$. All the points in $\mathcal{S}$ lie on $\mathcal{C}$, such that the angle at the tip formed by the two extreme points $L,R \in \mathcal{S}$, i.e.~$\angle{LTR} = 30^\circ$ in the anticlockwise direction. So, we have $\overline{TL} = \overline{TR} = r$. The distance between any two consecutive points in $\mathcal{S}$ is the same, say $d$. Let the number of points in $\mathcal{S}$ be $n$. We denote the Conic Set as $\mathrm{Cone}(T,r,n)$ and $\mathcal{S}$ as $\mathrm{Circ}(T,r,n)$. We call $TL$ as the left slope of the Conic Set and $TR$ as the right slope of the Conic Set.
\end{definition}

\begin{figure}[h] 
\centering
\subfloat[\centering $\Omega'_0$]{\includegraphics[width=5cm]{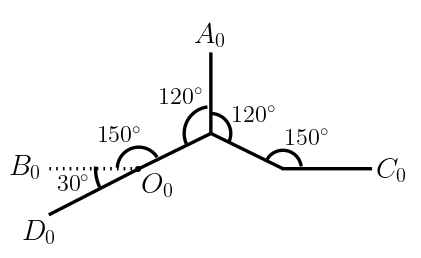}}   \qquad\qquad\qquad\qquad\qquad\qquad
\subfloat[\centering $\Omega'_1$]{\includegraphics[width=5cm]{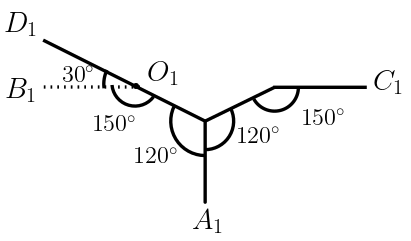}}
    \qquad\qquad
\subfloat[\centering $\Omega'_2$] {\includegraphics[width=5cm]{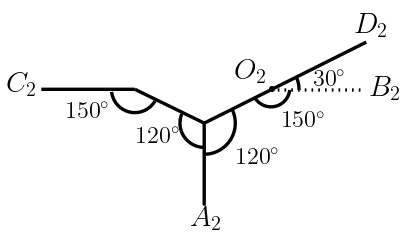} }
\caption{The modified terminator gadgets}  
\label{fig:terminator_new}
\end{figure}

We use the Conic Set in the reduction for our problem. Now, we state the reduction of an X3C instance $(n,\mathbb{F})$ to an instance $X'(\mathbb{F},\epsilon)$ of \ESMT. Later, we show that the instance will satisfy the desired properties on the number of terminals inside the convex hull of the terminal set.

\begin{description}
\item [Algorithm $\mathcal{A}$ for construction of an ESMT instance $X'(\mathbb{F},\epsilon)$ from an X3C instance $(n,\mathbb{F})$:]
\hspace{0.1cm}
\begin{itemize}
    \item Reduce the input X3C instance to the points configuration $X(\mathbb{F})$ according to the reduction given in~\cite{garey1977complexity}.
    \item Modify the terminators $\Omega_0$, $\Omega_1$, and $\Omega_2$ to as shown in~\Cref{fig:terminator_new} and call them $\Omega'_0$, $\Omega'_1$, and $\Omega'_2$. Let $DQCP$ be the smallest axis-parallel rectangle bounding $X(\mathbb{F})$ after modifying the terminators, where $D$ is the bottom leftmost point of $\Omega'_0$.
    \item Take $\alpha = \frac{1}{\epsilon}$. Define $r = ct^{\alpha} = \OO(t^{\alpha})$ and $n' = c't^{\gamma\alpha} = \OO(t^{\gamma\alpha})$, where $t=|\mathbb{F}|$ and $c$ and $c'$ are constants. Add the $\mathrm{Cone}(D,r,n')$, such that $D$ is the tip of the Conic Set, and the left slope $DE$ makes an angle of $120^\circ$ with $DP$. The right slope $DF$ also makes an angle of $120^\circ$ with $DQ$.
\end{itemize}
\end{description}

\begin{figure}[h]
\centering
\includegraphics[width=10cm]{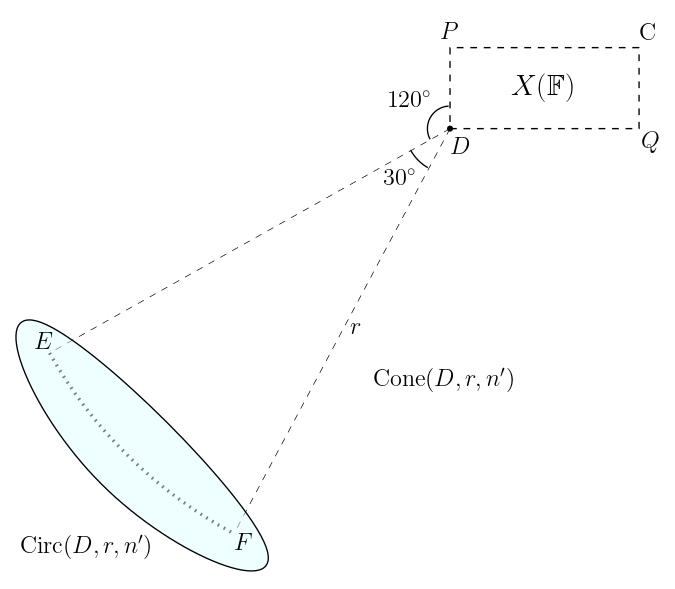}
\caption{The reduced instance $X'(\mathbb{F},\epsilon)$}
\label{fig:redn_instance}
\end{figure}


Now we prove a few properties of the constructed instance $X'(\mathbb{F},\epsilon)$.
\begin{lemma}\label{lem:convhull_pts}
    All the points in $\mathrm{Circ}(D,r,n')$ (according to~\Cref{def:conic_set}) lie on the convex hull of the reduced ESMT instance $X'(\mathbb{F},\epsilon)$ constructed by Algorithm $\mathcal{A}$, where $\epsilon \in (0,1]$.
\end{lemma}

\begin{proof}
    By the construction of $\mathrm{Cone}(D,r,n')$ in Algorithm $\mathcal{A}$, let $\mathcal{C}$ be the circle on which all the points in $\mathrm{Circ}(D,r,n')$ lie. If we draw a tangent to $\mathcal{C}$ at any of the points in $\mathrm{Circ}(D,r,n')$, then all the remaining points in the configuration $X'(\mathbb{F},\epsilon)$ lie towards one side of the tangent. We know that if we can find a line passing through a point such that all the other points in the plane lie on one side of the line, then the point lies on the convex hull of the points in the plane. Therefore, all the points in $\mathrm{Circ}(D,r,n')$ lie on the convex hull of the reduced instance $X'(\mathbb{F},\epsilon)$.
\end{proof}

Let us denote the convex hull of $X'(\mathbb{F},\epsilon)$ by $\mathrm{CH}(X'(\mathbb{F},\epsilon))$ and that of the points lying inside or on the bounding rectangle $\mathrm{PDQC}$, i.e.~$X'(\mathbb{F},\epsilon)\setminus \mathrm{Circ}(D,r,n')$ by $\mathrm{CH}(X'(\mathbb{F},\epsilon)\setminus \mathrm{Circ}(D,r,n'))$.

\begin{lemma}\label{lem:redcn_size}
    The reduced ESMT instance $X'(\mathbb{F},\epsilon)$ constructed by Algorithm $\mathcal{A}$ has $\Omega (N^\epsilon)$ points inside the convex hull, where $\epsilon \in (0,1]$ and $N$ is the total number of terminals in $X'(\mathbb{F},\epsilon)$. 
\end{lemma}

\begin{proof}
    $\mathrm{CH}(X'(\mathbb{F},\epsilon))$ contains all the points in $\mathrm{Circ}(D,r,n')$ by~\Cref{lem:convhull_pts}. $\mathrm{Circ}(D,r,n')$ contains $n' = \OO(t^{\gamma\alpha})$ points.

    Now we need to analyze the number of points on $\mathrm{CH}(X'(\mathbb{F},\epsilon)\setminus \mathrm{Circ}(D,r,n'))$. The remaining points in $X(\mathbb{F})$, i.e.~$X(\mathbb{F})\setminus \mathrm{CH}(X'(\mathbb{F},\epsilon)\setminus \mathrm{Circ}(D,r,n'))$ must lie within the convex hull of the entire construction, i.e.~$X'(\mathbb{F},\epsilon)$. From the construction in Algorithm $\mathcal{A}$, no point on the connecting rows can be a part of $\mathrm{CH}(X'(\mathbb{F},\epsilon)\setminus \mathrm{Circ}(D,r,n'))$ as there is no line passing through it, which contains all terminals on one side of it. The same thing holds for the square and hexagonal gadgets as well, except the hexagonal gadgets corresponding to the last element of the last set in the family, i.e.~$F_t$. Thus, only the terminators and the hexagonal gadgets corresponding to the last element of $F_t$ contribute points to $\mathrm{CH}(X'(\mathbb{F},\epsilon)\setminus \mathrm{Circ}(D,r,n'))$.

    If we look at the arrangement of points in the terminators (modified as well as those left unchanged) and the hexagonal gadgets as shown in~\Cref{fig:terminator,fig:hexagon_points}, the convex hull of each of these gadgets consists of constantly many points. Therefore, the number of points each of these gadgets contribute to $\mathrm{CH}(X'(\mathbb{F},\epsilon)\setminus \mathrm{Circ}(D,r,n'))$ is bounded by some constant. The number of terminators is $6t+2$ and the number of hexagonal gadgets corresponding to the last element of $F_t$ is at most $3n$. Therefore, the number of points on $\mathrm{CH}(X'(\mathbb{F},\epsilon)\setminus \mathrm{Circ}(D,r,n'))$ is $\OO(t+n) = \OO(t)$ as $n \leq t$.
    
    The instance $X(\mathbb{F})$ obtained via reduction from X3C has $6t+2$ terminators, $t$ squares, at most $9nt$ crossovers (hexagonal gadgets), and $\OO(nt)$ connecting rows of points. The number of gadgets is $\OO(nt)$. Therefore, the total number of points in $X(\mathbb{F})$ is $\Omega(nt) = \omega(t)$. The modified terminators result in a constantly many increase in the number of points. So, we have $\gamma > 1$.

    Thus, the number of points inside the convex hull is $\Omega(t^\gamma)$ and those on the convex hull is $\OO(t^{\gamma\alpha})$. So, the total number of terminals is $N = \OO(t^{\gamma\alpha}) + \OO(t^{\gamma}) = \OO(t^{\gamma\alpha})$, and those inside the convex hull is $\Omega(t^\gamma) = \Omega(N^{1/\alpha}) = \Omega(N^\epsilon)$ as $\alpha = \frac{1}{\epsilon}$.
\end{proof}

\begin{figure}[h]
\centering
\includegraphics[width=5cm]{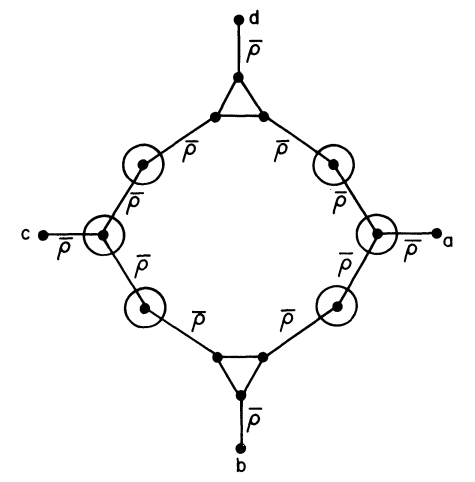}
\caption{The hexagonal gadget (crossover), the convex hull of the gadget is the quadrilateral $\rm abcd$ (taken from~\cite{garey1977complexity})}
\label{fig:hexagon_points}
\end{figure}


    

We further prove structural properties of SMTs of the reduced instance $X'(\mathbb{F},\epsilon)$ when considering the modified gadgets $\Omega'_0$, $\Omega'_1$, and $\Omega'_2$.
\begin{lemma}\label{lem:modified_smt}
    Consider an SMT $\mathcal{S}^*$ of the ESMT instance $X(\mathbb{F})$ obtained via reduction from the X3C instance $(n,\mathbb{F})$ as per~\cite{garey1977complexity}. Consider a tree $\mathcal{S'}^{*}$ on the terminal set of $X'(\mathbb{F},\epsilon)$ obtained from $\mathcal{S}^*$ as follows: Consider the modified terminator gadgets $\Omega'_i,~i \in \{0,1,2\}$ as in Algorithm $\mathcal{A}$. For each $i \in \{0,1,2\}$, the edge $B_iO_i$ is excluded from $\mathcal{S}^*$ and the edge $D_iO_i$ is included to form $\mathcal{S'}^{*}$. $\mathcal{S'}^{*}$ is an SMT for the terminal set of $X'(\mathbb{F},\epsilon)$.
\end{lemma}

\begin{proof}
    Consider $\mathcal{S}^*$. Due to Lemma 4 of~\cite{garey1977complexity}, Steiner points of $\mathcal{S}^*$ can only be connected to points in the triangular and square gadgets. Lemma 5 of~\cite{garey1977complexity} states that if there are two terminals $x,y \in X(\mathbb{F})$ and the distance between $x$ and $y$ does not exceed $\frac{1}{10}$, then $(x,y)$ is an edge of $\mathcal{S}^{*}$. So in $\mathcal{S}^{*}$, for each $i \in \{0,1,2\}$ all the points on $B_iO_i$ are joined together along $B_iO_i$. Lemma 5 of~\cite{garey1977complexity} also holds true on modifying the terminators to $\Omega'_i,~i \in \{0,1,2\}$. Now, we join all the points on $D_iO_i$ along $D_iO_i$. This gives us the SMT $\mathcal{S'}^*$ for the terminal set of $X'(\mathbb{F},\epsilon)$.
\end{proof}


Now we focus on the structure of the SMT of $X'(\C{F},\epsilon)$. The SMT is basically the union of the SMT $\C{S'}^{*}$ of the points in the bounding rectangle $PDQC$ as stated in~\Cref{lem:modified_smt} and the SMT of the set of points $\mathrm{Cone}(D,r,n')$.

$\mathrm{CH}(X'(\mathbb{F},\epsilon)\setminus \mathrm{Circ}(D,r,n'))$ is enclosed by the bounding rectangle $PDQC$ and $D$ must lie on $\mathrm{CH}(X'(\mathbb{F},\epsilon)\setminus \mathrm{Circ}(D,r,n'))$. We label the vertices of $\mathrm{CH}(X'(\mathbb{F},\epsilon)\setminus \mathrm{Circ}(D,r,n'))$ as $D,P_1,P_2,\ldots,P_k$ in the counter-clockwise order. Let $\mathrm{CH}(X'(\mathbb{F},\epsilon))$ be the convex hull of all the points. By~\Cref{lem:convhull_pts}, all the points in $\mathrm{Circ}(D,r,n')$ lie on $\mathrm{CH}(X'(\mathbb{F},\epsilon))$. Let $EP_i$ and $FP_j$ be edges in $\mathrm{CH}(X'(\mathbb{F},\epsilon))$, such that $P_i,P_j \notin \mathrm{Circ}(D,r,n')$.

The SMT of $X'(\mathbb{F},\epsilon)$ clearly lies inside its convex hull, $\mathrm{CH}(X'(\mathbb{F},\epsilon))$. We show that the Steiner hull can be further restricted to the bounding rectangle $PDQC$ and the convex polygon formed by the points in $\mathrm{Cone}(D,r,n')$. For this we use Theorem 1.5 in~\cite{hwang1992steiner}, as stated below.

\begin{proposition}~\cite{hwang1992steiner}\label{thm:steiner_hull}
    Let $H$ be a Steiner hull of $N$. By sequentially removing wedges $\rm abc$ from the remaining region, where $\rm a$, $\rm b$, $\rm c$ are terminals but $\triangle{\rm abc}$ contains no other terminal, $a$ and $c$ are on the boundary and $\angle{\rm abc} \geq 120^\circ$, a Steiner hull $H'$ invariant to the sequence of removal is obtained.
\end{proposition}

\begin{lemma}\label{thm:steiner_hull_regions}
    The region comprising of the bounding rectangle $PDQC$ according to Algorithm $\mathcal{A}$ and the convex polygon formed by the set of points $\mathrm{Cone}(D,r,n')$ is a Steiner hull of $X'(\mathbb{F},\epsilon)$.
\end{lemma}

\begin{proof}
    Firstly, let us consider the wedge $EP_{i+1}P_i$. All the points are terminals, $E$ and $P_i$ are boundary points, and $\triangle{EP_{i+1}P_i}$ contains no other terminal. Now, $\angle{EP_{i+1}P_i}$ is greater than the exterior angle of $\angle{EP_{i+1}D}$, which in turn is greater than $\angle{EDP_{i+1}}$. $\angle{EDP_{i+1}} \geq \angle{\rm EDP} = 120^\circ$, by the construction. Therefore, $\angle{EP_{i+1}P_i} \geq 120^\circ$. By applying~\Cref{thm:steiner_hull}, we can remove the wedge $EP_{i+1}P_i$ from the convex hull $\mathrm{CH}(X'(\mathbb{F},\epsilon))$ to get a smaller Steiner hull. This can be repeated for the wedges $EP_{i+2}P_{i+1}, EP_{i+3}P_{i+2}, \ldots, EDP_k$. The same argument can also be used to get rid of the wedges $FP_{j-1}P_j, FP_{j-2}P_{j-1}, \ldots, FDP_1$. So, we get the final Steiner hull $H'$ to be the union of the bounding rectangle $PDQC$ and the convex polygon formed by the points in $\mathrm{Cone}(D,r,n')$.
\end{proof}

Given the nature of the above Steiner hull, we show that we can treat $X(\mathbb{F})$ and $\mathrm{Cone}(D,r,n')$ separately. 
\begin{lemma}\label{thm:steiner_tree}
    There is an SMT of $X'(\mathbb{F},\epsilon)$ that is the union of an SMT of $X(\mathbb{F})$ and an SMT of the points in $\mathrm{Cone}(D,r,n')$, with $D$ being common to both of them.
\end{lemma}

\begin{proof}
    According to~\Cref{thm:steiner_hull_regions}, there is an SMT of $X'(\mathbb{F},\epsilon)$ that lies completely inside the the bounding quadrilateral $PDQC$ and the convex polygon formed by $\mathrm{Cone}(D,r,n')$. These two regions have $D$ as the only common point. Therefore, $D$ is an articulation point in the tree and connects these two regions. So, we have this SMT of $X'(\mathbb{F},\epsilon)$ as the union of an SMT of $X(\mathbb{F})$ and an SMT of the points in $\mathrm{Cone}(D,r,n')$.
\end{proof}

We can identify a structure for an SMT of the points in $\mathrm{Cone}(D,r,n')$ using~\cite{weng1995steiner}.
\begin{lemma}\label{thm:steiner_tree_part}
    There is an SMT of the points in $\mathrm{Cone}(D,r,n')$ that is as shown in~\Cref{fig:steiner_new}. In the SMT, $D$ is connected to the two middle points in $\mathrm{Circ}(D,r,n')$ via a Steiner point $S^t$. The other points in $\mathrm{Circ}(D,r,n')$ are connected along the circumference.
\end{lemma}

\begin{proof}
    The number of points in $\mathrm{Circ}(D,r,n')$ is $\OO(t^{\gamma\alpha})$. We can take the constant factor to be large enough so that $|\mathrm{Circ}(D,r,n')| >= 12$. If we complete the regular polygon on $\mathcal{C}$ having $\mathrm{Circ}(D,r,n')$ as a subset of its vertices, then it contains more than $12$ vertices and along with the centre $D$ has a SMT with structure given in~\cite{weng1995steiner}.

    Let the Steiner tree for $\mathrm{Cone}(D,r,n')$ as shown in~\Cref{fig:steiner_new} be denoted by $\mathcal{T}_1$. If this is not minimal, then there exists another Steiner tree $\mathcal{T}_2$ such that $|\mathcal{T}_2| < |\mathcal{T}_1|$. Then we can replace $\mathcal{T}_1$ by $\mathcal{T}_2$ in the SMT of the regular polygon and its centre to get a shorter Steiner tree. This contradicts the minimality of the structure given in~\cite{weng1995steiner}. Therefore, the SMT of $\mathrm{Cone}(D,r,n')$ follows the structure in~\Cref{fig:steiner_new}.
\end{proof}

\begin{figure}[h]
\centering
\includegraphics[width=8cm]{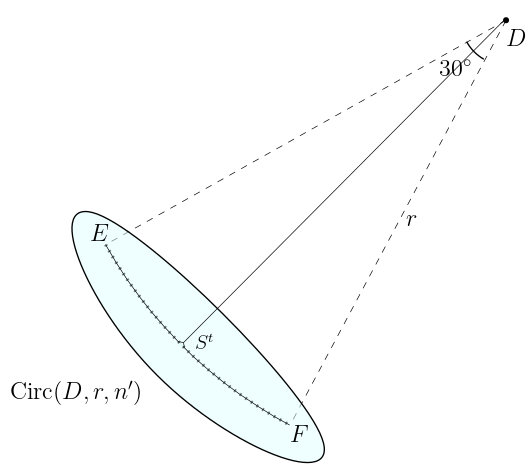}
\caption{SMT of $\mathrm{Cone}(D,r,n')$}
\label{fig:steiner_new}
\end{figure}
Finally, we prove the NP-hardness of \ESMT on $f(n)$-Almost Convex Sets of $n$ terminals, when $f(n) = \Omega(n^\epsilon)$ for some $\epsilon \in (0,1]$.
\begin{theorem}\label{thm:redn_final}
    Let $\C{S}^{*}_{\mathbb{F},\epsilon}$ denote an SMT of $X'(\mathbb{F},\epsilon)$ and $|\C{S}^*_{\mathbb{F},\epsilon}|$ denote its length. If $\mathbb{F}$ has an exact cover, then $|\C{S}^{*}_{\mathbb{F},\epsilon}| \leq f(n,t,\hat{C}) + |\mathcal{T}_1|$, otherwise $|\C{S}^{*}_{\mathbb{F},\epsilon}| \geq f(n,t,\hat{C}) + \frac{1}{200nt} + |\mathcal{T}_1|$, where $\hat{C}$ is the number of crossovers, i.e.~hexagonal gadgets, and $f$ is a positive real-valued function of $n,t,\hat{C}$ as stated in~\Cref{thm:redcn}.
\end{theorem}

\begin{proof}
    From~\Cref{thm:steiner_tree}, we have $|\C{S}^{*}_{\mathbb{F},\epsilon}| = |\C{S}^{*}| + |\mathcal{T}_1|$. From~\Cref{thm:steiner_tree_part}, we can compute the length of $\mathcal{T}_1$ as a function of $t$, $\alpha$, and $\gamma$. Finally, using~\Cref{thm:redcn} we get the required reduction.
\end{proof}




Since it is not known if the ESMT problem is in NP, Garey et al.~\cite{garey1977complexity} show the NP-completeness of a related problem called the {\sc Discrete Euclidean Steiner Minimal Tree} (DESMT) problem, which is in NP. We define the DESMT problem as given in~\cite{garey1977complexity}. The DESMT problem takes as input a set $\C{X}$ of integer-coordinate points in the plane and a positive integer $L$, and asks if there exists a set $\C{Y} \supseteq \C{X}$ of integer-coordinate points such that some spanning tree $\C{T}$ for $\C{Y}$ satisfies $|\C{T}|_d \leq L$, where $|\C{T}|_d = \Sigma_{e \in E(\mathcal T)} \lceil\overline{e}\rceil$, i.e.~we round up the length of each edge to the least integer not less than it.

In order to show that DESMT is NP-hard, the same reduction as that of the ESMT problem can be used, followed by scaling and rounding the coordinates of the points. Theorem 4 of~\cite{garey1977complexity} proves that the DESMT problem is NP-Complete. Moreover, since it is Strongly NP-Complete, the DESMT problem does not admit any FPTAS. Finally in Theorem 5 of~\cite{garey1977complexity}, Garey et al. show that as a consequence, the ESMT problem does not have any FPTAS as well.

Now we show that the DESMT problem is NP-hard even on $f(n)$-Almost Convex Sets of $n$ terminals, when $f(n) = \Omega (n^\epsilon)$ and where $\epsilon \in (0,1]$.

In Section 7 of~\cite{garey1977complexity}, the reduced instance $X(\mathbb{F})$ of ESMT is converted into an instance $X_{d}(\mathbb{F})$ of DESMT. The conversion is as follows:\\
$X_{d}(\mathbb{F}) = \{(\lceil 12M\cdot 200nt\cdot x_1\rceil, \lceil 12M\cdot 200nt\cdot x_2\rceil): x=(x_1,x_2)\in X(\mathbb{F})\}$, where $M = |X(\mathbb{F})|$.

We apply a similar conversion to the reduced ESMT instance $X'(\mathbb{F},\epsilon)$, to convert it into a DESMT instance of an $\Omega(n^\epsilon)$-Almost Convex Set. The conversion goes as follows:\\
$X'_{d}(\mathbb{F},\epsilon) = \{(\lceil 12N\cdot 200nt\cdot x_1\rceil, \lceil 12N\cdot 200nt\cdot x_2\rceil): x=(x_1,x_2)\in X'(\mathbb{F},\epsilon)\}$, where $N = |X'(\mathbb{F},\epsilon)|$.

The next two lemmas establish the validity of $X'_{d}(\mathbb{F},\epsilon)$ as an instance of DESMT and the upper bounds on the size of the constructed instance. Note that the reduction from X3C followed by the conversion can be done in polynomial time.
\begin{lemma}\label{lem:valid_desmt}
    The instance $X'_{d}(\mathbb{F},\epsilon)$ constructed above is a valid DESMT instance.
\end{lemma}

\begin{proof}
    All the points in $X'_{d}(\mathbb{F},\epsilon)$ have integer coordinates according to the conversion stated above. So, it is a DESMT instance.
\end{proof}

\begin{lemma}\label{lem:numpts_same}
    The reduced DESMT instance $X'_{d}(\mathbb{F},\epsilon)$ has $N$ distinct points, where $N = |X'(\mathbb{F},\epsilon)|$.
\end{lemma}

\begin{proof}
    The minimum distance between any two points in $X'(\mathbb{F},\epsilon)$ is that between two consecutive points of $\mathrm{Circ}(D,r,n')$, which is $\OO(\frac{1}{t^\gamma})$. Recall~\Cref{lem:redcn_size} which establishes that $N = \OO(t^{\gamma\alpha})$. So, the minimum distance between any two points in $X'_{d}(\mathbb{F},\epsilon)$ is $\OO(N \cdot nt \cdot \frac{1}{t^\gamma}) = \OO(nt^{\alpha+1})$. Because of the substantial distance obtained between points after scaling, the rounding will not cause any distinct points of $X'_{d}(\mathbb{F},\epsilon)$ to coincide. Therefore, the number of points remains unchanged, i.e.~$|X'_{d}(\mathbb{F},\epsilon)| = |X'(\mathbb{F},\epsilon)| = N$.
\end{proof}

Now we present the following lemma for the constructed DESMT instance $X'_{d}(\mathbb{F},\epsilon)$ analogous to~\Cref{lem:redcn_size} for the ESMT instance $X'(\mathbb{F},\epsilon)$.

\begin{lemma}\label{lem:desmt_redn_size}
    The reduced DESMT instance $X'_{d}(\mathbb{F},\epsilon)$ constructed is an $\Omega(N^\epsilon)$-Almost Convex Set, where  $N = |X'_{d}(\mathbb{F},\epsilon)|$.
\end{lemma}

\begin{proof}
    From~\Cref{lem:redcn_size}, we know that the reduced ESMT instance $X'(\mathbb{F},\epsilon)$ has $\Omega(N^\epsilon)$ points inside its convex hull $\mathrm{CH}(X'(\mathbb{F},\epsilon))$, and $N = |X'(\mathbb{F},\epsilon)|$. The number of points after conversion remains the same by~\Cref{lem:numpts_same}. We need to show that after conversion, except for the points of $\OO(t)$ gadgets, no other points inside the convex hull $\mathrm{CH}(X'(\mathbb{F},\epsilon))$ lie on the new convex hull $\mathrm{CH}(X'_{d}(\mathbb{F},\epsilon))$. The number of points in each of the $\OO(t)$ anomalous gadgets are constant in number, and hence not too many points from the interior of $\mathrm{CH}(X'(\mathbb{F},\epsilon))$ can lie on $\mathrm{CH}(X'_{d}(\mathbb{F},\epsilon))$. 
    
    After conversion, all the points on a horizontal connecting row have the same $y$-coordinate, as they initially had the same $y$-coordinate and therefore undergo the same transformation. Thus, all the points on a horizontal connecting row still lie on a horizontal line segment in $X'_{d}(\mathbb{F},\epsilon)$. Similarly, all the points on a vertical connecting row still lie on a vertical line segment in $X'_{d}(\mathbb{F},\epsilon)$. This implies that none of the points on the connecting rows can be a part of $\mathrm{CH}(X'_{d}(\mathbb{F},\epsilon))$ as there can be no line passing through them that also contains all terminal points on one side of it. 

    The same thing holds for the square and hexagonal gadgets (crossovers) as well, except the hexagonal gadgets placed at the beginning or end of any row. This is because all the points which are a part of these square and hexagon gadgets are surrounded by connecting row points all four sides, above, below, left and right. So again, only the terminators and the hexagonal gadgets appearing at the beginning or end of any row contribute to $\mathrm{CH}(X'_{d}(\mathbb{F},\epsilon))$.

    Now, since we had adjusted the number of points in the long rows of the terminators and hexagonal gadgets such that their lengths and breadths are some constants, the number of points in each of the terminators and hexagonal gadgets can be bounded by some constant as the minimum distance between any two consecutive points on the long rows or standard rows is at least $\frac{1}{11}$. Therefore, each of these gadgets contribute some constantly many points to $\mathrm{CH}(X'_{d}(\mathbb{F},\epsilon))$.
    
    As we have seen in~\Cref{lem:redcn_size}, the number of terminators is $6t+2$ and the number of hexagonal gadgets corresponding to the beginning or end of any row is at most $6n$. Therefore, the number of points contributed by the terminators and the hexagonal gadgets placed at the beginning or the end of any row, to $\mathrm{CH}(X'_{d}(\mathbb{F},\epsilon)$ is $\OO(t+n) = \OO(t)$, as $n \leq t$. Even if all the points in $\mathrm{Circ}(D,r,n')$ lie on the new convex hull $\mathrm{CH}(X'_{d}(\mathbb{F},\epsilon)$, we have $\Omega(t^\gamma) = \Omega(N^\epsilon)$ points inside it. Thus we are done.
\end{proof}

We get the following theorem from~\Cref{lem:valid_desmt,lem:numpts_same,lem:desmt_redn_size}.

\begin{theorem}\label{thm:desmt_redn}
    The instance $X'_{d}(\mathbb{F},\epsilon)$ constructed is a valid DESMT instance on an $\Omega(N^\epsilon)$-Almost Convex Set, where  $|X'_{d}(\mathbb{F},\epsilon)| = |X'(\mathbb{F},\epsilon)| = N$.
\end{theorem}

Following Theorems 3 and 4 in~\cite{garey1977complexity}, we get that the DESMT problem is NP-Complete for $\Omega(N^\epsilon)$-Almost Convex Sets, where $N$ is the total number of terminals. Since we get the reduced instance $X'_{d}(\mathbb{F},\epsilon)$ from the X3C instance $(n,\mathbb{F})$, the DESMT problem is strongly NP-complete for $\Omega(N^\epsilon)$-Almost Convex Sets, and does not admit any FPTAS.

Using Theorem 5 of~\cite{garey1977complexity}, we get that if the ESMT problem has an FPTAS, then the X3C problem can be solved in polynomial time. The Theorem also applies for our case of $\Omega(N^\epsilon)$-Almost Convex Sets. Therefore, we get the following theorem,

\begin{theorem}\label{thm:no_fptas}
    There does not exist any FPTAS for the ESMT problem on $f(n)$-Almost Convex Sets of $n$ terminals, where $f(n) = \Omega(n^\epsilon)$ and $\epsilon \in (0,1]$, unless \pnp.
\end{theorem}

\section{Conclusion}
 In this paper, we first study ESMT on vertices of $2$-CPR $n$-gons and design a polynomial time algorithm. It remains open to design a polynomial time algorithm for ESMT on $k$-CPR $n$-gons, or show NP-hardness for the problem. Next, we study the problem on $f(n)$-Almost Convex Sets of $n$ terminals. For this NP-hard problem, we obtain an algorithm that runs in $2^{\OO(f(n)\log n)}$ time. We also design an FPTAS when $f(n) = \OO(\log n)$. On the other hand, we show that there cannot be an FPTAS if $f(n) = \Omega(n^{\epsilon})$ for any $\epsilon \in (0,1]$, unless \pnp. The question of existence of FPTASes when $f(n)$ is a polylogarithmic function remains open.

\bibliography{refs}

\end{document}